\def\llncs{0}
\def\issubmission{0}
\def\comments{0}

\ifnum\llncs=1
\documentclass{llncs}
\else
\documentclass[11pt]{article}
\usepackage{fullpage}
\fi
\usepackage{amsmath}

\usepackage{amsthm}
\usepackage{amssymb}
\usepackage{physics}
\usepackage{hyperref}
\usepackage{cleveref}
\usepackage{dsfont}
\usepackage{xcolor}
\usepackage{mathtools}
\usepackage{breakcites}
\usepackage{comment}
\usepackage{tikz}
\usetikzlibrary{quantikz}
\usepackage{bm}
\usepackage{enumitem}

\ifnum\llncs=0
\newtheorem{theorem}{Theorem}[section]

\newtheorem{corollary}[theorem]{Corollary}
\newtheorem{definition}[theorem]{Definition}
\newtheorem{lemma}[theorem]{Lemma} 
\newtheorem{claim}{Claim}      

\theoremstyle{definition}

\theoremstyle{remark}
\newtheorem{remark}[theorem]{Remark}
\newtheorem{note}[theorem]{Note}

\fi

\newcommand{\A}{\mathcal{A}}

\newcommand{\N}{\mathbb{N}}

\newcommand{\E}{\mathop{\mathbb{E}}}

\newcommand{\wt}[1]{\widetilde{#1}}

\DeclareMathOperator{\Sim}{Sim}

\newcommand{\secparam}{\lambda}
\newcommand{\regA}{{\color{gray} {\bf A}}}
\newcommand{\regB}{{\color{gray} {\bf B}}}
\newcommand{\regC}{{\color{gray} {\bf C}}}
\newcommand{\regD}{{\color{gray} {\bf D}}}

\newcommand{\regK}{{\color{gray} {\bf K}}}

\newcommand{\regM}{{\color{gray} {\bf M}}}

\newcommand{\regR}{{\color{gray} {\bf R}}}
\newcommand{\regS}{{\color{gray} {\bf S}}}

\newcommand{\regX}{{\color{gray} {\bf X}}}
\newcommand{\regY}{{\color{gray} {\bf Y}}}
\newcommand{\regZ}{{\color{gray} {\bf Z}}}
\newcommand{\regIn}{{\color{gray} {\bf In}}}
\newcommand{\regOut}{{\color{gray} {\bf Out}}}
\newcommand{\regSt}{{\color{gray} {\bf St}}}
\newcommand{\iso}{\mathsf{iso}}

\newcommand{\PureV}{\mathsf{PureV}}
\newcommand{\regAnc}{{\color{gray} {\bf Anc}}}

\newcommand{\hilbert}{{\cal H}}
\newcommand{\Unitary}{{\cal U}}
\renewcommand{\Tr}{\mathsf{Tr}}
\newcommand{\TD}{\mathsf{TD}}
\newcommand{\negl}{\mathsf{negl}}
\newcommand{\poly}{\mathsf{poly}}
\newcommand{\type}{\mathsf{type}}
\newcommand{\types}{\mathsf{TYPES}}
\newcommand{\setst}{\mathsf{Set}}
\newcommand{\expect}{\mathbb{E}}

\newcommand{\distr}{\mathcal{D}}
\newcommand{\hybrid}{\mathsf{Hybrid}}
\newcommand{\adv}{{\cal A}}
\newcommand{\reduction}{{\cal R}}

\newcommand{\gen}{{\sf Gen}}
\newcommand{\mint}{{\sf Mint}}
\newcommand{\verify}{{\sf Verify}}
\newcommand{\signsetup}{{\cal K}{\cal G}}
\newcommand{\sign}{{\cal S}}
\newcommand{\signver}{{\cal V}}
\newcommand{\sk}{{\sf sk}}
\newcommand{\vk}{{\sf vk}}
\newcommand{\adversary}{{\cal A}}
\newcommand{\prob}{{\sf Pr}}

\newcommand{\U}{\mathcal{U}}
\newcommand{\myvdots}{\raisebox{1ex}{\vdots}}

\usepackage{mathrsfs}

\newcommand{\Apply}{\mathsf{Apply}}
\newcommand{\Select}{\mathsf{Sel}}
\newcommand{\sims}{{\sf Sim}}

\newcommand{\distinct}{{\sf dis}}

\newcommand{\Haar}{\mathscr{H}}
\newcommand{\pk}{\mathsf{pk}}

\newcommand{\PRU}{PRU}

\newcommand{\ESim}{\Sim_{\mathsf{eff}.\mathsf{iso}}}
\newcommand{\expand}{\mathsf{Expand}}
\newcommand{\Img}{\mathsf{Im}}

\newcommand{\BigPr}[2]{
\Pr\left[
\begin{array}{c}
#1
\end{array}
:
\begin{array}{c}
#2
\end{array}
\right]
}

\ifnum\comments=1
\newcommand{\todo}[1]{{\color{blue} (TODO: #1) }}
\newcommand{\enote}[1]{{\footnotesize {\color{purple} E:\ #1}}}
\newcommand{\pnote}[1]{{\footnotesize {\color{red} P:\ #1}}}
\fi

\ifnum\comments=0
\newcommand{\todo}[1]{}
\newcommand{\enote}[1]{}
\newcommand{\pnote}[1]{}
\fi

\title{Less is More:\\ On Copy Complexity in Quantum Cryptography\ifnum\issubmission=0\thanks{Work done in part while visiting the Simons Institute for the Theory of Computing, Berkeley.}\fi}
\ifnum\issubmission=0
\author{Prabhanjan Ananth\thanks{prabhanjan@cs.ucsb.edu}\\ UCSB \and Eli Goldin\thanks{eli.goldin@nyu.edu}\\ NYU}
\date{}
\else
\author{}
\fi
\ifnum\llncs=1
\institute{}
\fi

\begin{document}

\maketitle

\begin{abstract}
\noindent Quantum cryptographic definitions are often sensitive to the number of copies of the cryptographic states revealed to an adversary. Making definitional changes to the number of copies accessible to an adversary can drastically affect various aspects including the computational hardness, feasibility, and applicability of the resulting cryptographic scheme. This phenomenon appears in many places in quantum cryptography, including quantum pseudorandomness and unclonable cryptography. 

To address this, we present a generic approach to boost single-copy security to multi-copy security and apply this approach to many settings. As a consequence, we obtain the following new results: 
\begin{itemize}
    \item One-copy stretch pseudorandom state generators (under mild assumptions) imply the existence of $t$-copy stretch pseudorandom state generators, for any fixed polynomial $t$. 
    \item One-query pseudorandom unitaries with short keys (under mild assumptions) imply the existence of $t$-query pseudorandom unitaries with short keys, for any fixed polynomial $t$. 
    \item Assuming post-quantum pseudorandom functions exist, i.i.d.-copy secure uncloneable primitives imply identical-copy secure uncloneable primitives. In other words, many-time secure uncloneable primitives with mixed state output imply the same primitives with pure state output. This gives the first constructions of identical-copy secure public-key quantum money and quantum copy-protection.
\end{itemize}
\end{abstract} 




\ifnum\llncs=0
\newpage 

\tableofcontents
\newpage 
\fi

\section{Introduction}
\noindent The foundational principles of quantum mechanics impose constraints that force us to revisit cryptographic definitions and security models when designing quantum cryptographic primitives. Concretely, when formulating security definitions, due to the no-cloning principle of quantum mechanics, one has to be careful about the number of copies of the cryptographic quantum states that the adversary receives. While this may appear merely semantic at first glance, the copy complexity measure has important implications. Indeed, making definitional changes to the number of copies accessible to an adversary can drastically affect various aspects including the computational hardness, feasibility and applicability of the resulting cryptographic scheme. There are many such examples in quantum cryptography and we will touch upon a few below. 

\paragraph{Case study: Quantum Pseudorandomness and Microcrypt.} Pseudorandom state generators (PRSGs) are efficient algorithms that on input a classical key produce states that are computationally indistinguishable from Haar random states. That is, any computationally bounded adversary cannot tell apart whether it receives as input pseudorandom states or Haar random states. A natural question that arises in this definition is: how many copies of the state does the adversary receive? The original work~\cite{JLS18} that introduced pseudorandom states proposed a definition where the adversary can receive {\em a priori unbounded polynomial number} of copies. A more recent work~\cite{MY22} proposed a different definition referred to as {\em stretch} PRSGs wherein the adversary only receives one copy of the state and the key length is smaller than the number of qubits of the state. In the past few years, there have been several works that have demonstrated that these two definitions are vastly different.  
\par Recently,~\cite{CCS24} showed a separation between single-copy and multi-copy PRSGs. Moreover, it is believed that single-copy stretch PRSGs cannot be broken using any classical oracle~\cite{LMW24} while on the other hand, multi-copy PRSGs can be broken using a PP oracle~\cite{Kretschmer21,GMMY24}. Finally, there are cryptographic primitives, such as quantum pseudo one-time pads, that can be built from multi-copy secure PRSGs but not known from stretch PRSGs~\cite{AQY21}. 
\par Broadly speaking, a pseudorandom state generator is just one notion in the expanding world of microcrypt, which is comprised of a variety of quantum primitives that are believed to exist even if one-way functions don't. Another popular resident of this world is a one-way state generator. It was shown by~\cite{cavalar2023computational} that a one-way state generator with $n$-qubit output can be realized with information-theoretic security if the number of copies received by the adversary is $o(\frac{n}{\log(n)})$. This is tight since it was shown by~\cite{KT24,BJ24} that $\omega\left( \frac{n}{\log(n)} \right)$-copy secure one-way state generators imply the existence of quantum bit commitments.

\paragraph{Case Study: Unclonable Cryptography.} Unclonable cryptographic primitives are yet another set of primitives where the number of copies of the unclonable state the adversary receives is critical in the security definition. Most of the unclonable notions studied in the literature only guarantee security if the adversary only receives one copy of the unclonable state. In fact, some of the constructions are easily broken if the adversary receives multiple copies of the unclonable state. The first work to study multi-copy security for unclonable primitives was Aaronson~\cite{Aaronson18} who argued that in some settings, using shadow tomography, many unclonable primitives can be broken if the adversary receives many copies of the state. Following Aaronson, several recent works~\cite{LLQZ22,CG24b,AMP24,KNP25,PRV24} attempt to show the feasibility of multi-copy security for a limited number of unclonable primitives, including copy-protection, single-decryptor encryption and revocable encryption. As demonstrated in these works, achieving multi-copy security turns out to be much harder than single-copy security. 

\paragraph{Our Work.} We set out to understand the copy complexity for many quantum cryptographic primitives. Specifically, we set out to understand the following question: 
\begin{quote}
\begin{center}
{\em In which settings does single-copy security imply multi-copy security?} 
\end{center} 
\end{quote}

\noindent We also study similar questions for cryptographic unitaries. As a concrete example, we study the relationship between pseudorandom unitaries (PRUs) secure against 1-query adversaries versus PRUs secure against adversaries that make polynomially many queries.

\subsection{Results} 
\noindent We show that indeed in many settings, single-copy security does imply multi-copy security. To prove this, we present the main theorem that reduces multi-copy security to single-copy security and then we show how to apply this general theorem for various applications. Intuitively, our main theorem will show that for every ensemble of mixed states $\{\sigma_i\}$, there exists an ensemble of purifications $\{\ket{\psi_k}\}$ such that $\ket{\psi_k}^{\otimes t}$ for a random $k$ reveals no more information than $(\sigma_{i_1}\otimes \dots \otimes \sigma_{i_t})$ for random $i_1,\dots,i_t$. Formally, we show that $\ket{\psi_k}^{\otimes t}$ can be \text{efficiently simulated} given $(\sigma_{i_1}\otimes \dots \otimes \sigma_{i_t})$. And so, in most applications with mixed state output, the mixed state can be replaced with a random sample from this purification ensemble.
\par We first explain the intuition behind the main theorem. We will begin by considering a simplified setting. Suppose there is a family of states $\{\ket{\phi_i}_{\regA}\}_{i \in \{0,1\}^n}$ supported on register $\regA$. This family could correspond to pseudorandom states, quantum money states and so on. Consider the state $\ket{\psi_{f_1,f_4}} = 2^{-\frac{n}{2}} \sum_{i} \omega_{2^n}^{f_1(i)}  \ket{\phi_{f_4(i)}}_{\regA} \ket{i}_{\regC}$, where $f_1,f_4$ are functions. That is, $\ket{\psi_{f_1,f_4}}$ is a uniform superposition of all the states  $\{\ket{\phi_i}_{\regA}\}_{i \in \{0,1\}^n}$ with a random phase. Then, our (simplified) main theorem states that $t$ copies of $\ket{\psi_{f_1,f_4}}$, where $f_1,f_4$ are random functions, can be efficiently simulated by having $t$ i.i.d copies of $\{\ket{\phi_i}_{\regA}\}_{i \in \{0,1\}^n}$. That is, $t$ copies of $\ket{\psi_{f_1,f_4}}$ can be simulated given $(\ket{\phi_{i_1}}_{\regA},\ldots,\ket{\phi_{i_t}}_{\regA})$, where $i_j$ is sampled uniformly at random. This means that if the underlying family $\{\ket{\phi_i}_{\regA}\}_{i \in \{0,1\}^n}$ satisfies i.i.d copy security, i.e., security holds even given independent copies from $\{\ket{\phi_i}_{\regA}\}_{i \in \{0,1\}^n}$ then the security also holds even given $t$ copies of the pure state $\ket{\psi_{f_1,f_4}}$. 
\par However, in some applications, the cryptographic state could either be mixed or the ancilla register could be traced out before it is revealed to the adversary. In this case, we generalize the above intuition as follows: this time, let the family be $\{\ket{\phi_i}_{\regA \regB}\}_{i \in \{0,1\}^n}$. Imagine cryptographic settings where only the register $\regA$ is revealed and in particular, $\regB$ is traced out. We update the above intuition by applying a quantum one-time pad on $\regB$ controlled on the register $\regC$. Specifically, we consider the state  $\ket{\psi_{f_1,f_2,f_3,f_4}}=2^{-\frac{n}{2}} \sum_{i} \omega_{2^n}^{f_1(i)}  (I \otimes X_{\regB}^{f_2(i)} Z_{\regB}^{f_3(i)})\ket{\phi_{f_4(i)}}_{\regA \regB} \ket{i}_{\regC}$. We similarly argue that $t$ copies of the state $\ket{\psi_{f_1,f_2,f_3,f_4}}=2^{-\frac{n}{2}}\sum_{i} \omega_{2^n}^{f_1(i)}  (I \otimes X_{\regB}^{f_2(i)} Z_{\regB}^{f_3(i)})\ket{\phi_{f_4(i)}}_{\regA \regB} \ket{i}_{\regC}$, where $f_1,\ldots,f_4$ are random functions, can be approximately simulated given $t$ i.i.d copies from $\{\Tr_{\regB}\left( \ketbra{\phi_i}_{\regA \regB} \right)\}_i$. \par We state the main theorem in more detail below.  

\begin{theorem}[Main Theorem; Informal]
    \label{thm:intro:mainthm}
    Consider a family of mixed states $\{\sigma_i\}_{i\in [N]}$ with corresponding purifications $\{\ket{\phi_i}\}_{i\in [N]}$. Then for all $n$, there exists a family of states $\{\ket{\psi_{k}}_{\regA,\regB}\}_{k\in \mathcal{K}}$ such that the following hold
    \begin{enumerate}
        \item $t$ copies of a random sample $\ket{\psi_k}^{\otimes t}$ can be efficiently approximated using $t$ i.i.d. copies of $\sigma_i$. That is, there exists an efficient simulator that gets as input $\E_{i_1,\dots,i_t}[\sigma_{i_1}\otimes \dots \otimes \sigma_{i_t}]$ and produces a state that is $\frac{t^2}{2^n}$ close (in trace distance) to $\E_{k}[\ketbra{\psi_k}^{\otimes t}]$. Furthermore, the simulator is universal and independent of the family $\{\sigma_i\}$.
        \item $\{\ket{\psi_k}\}_{k\in \mathcal{K}}$ is a purification of $\{\sigma_i\}_{i\in \{0,1\}^n}$. In particular, 
        $$\E_k[\Tr_{\regB}(\ketbra{\phi_k})] = \E_i[\sigma_i]$$
        For the case when $N=1$ and $\{\sigma_i\}$ consists of a single mixed state $\sigma$, we further have that for all $k\in \mathcal{K}$,
        $$\Tr_{\regB}[\ketbra{\phi_k}] = \sigma$$
        \item $\mathcal{K}$ is the set of functions $f_1,f_2,f_3,f_4$ on an appropriate domain and codomain. Given query access to $f_1,f_2,f_3,f_4$ and a circuit preparing $\{\ket{\phi_i}\}_{i\in \{0,1\}^n}$, $\ket{\psi_{f_1,f_2,f_3,f_4}}$ is efficiently preparable.
    \end{enumerate}
\end{theorem}

In particular, the construction of $\{\ket{\psi_{f_1,f_2,f_3,f_4}}\}$ will be
$$\ket{\psi_{f_1,f_2,f_3,f_4}} = \sum_{i} \frac{\omega_{2^n}^{f_1(i)}}{\sqrt{2^{n}}}  \left(I_{\regA} \otimes X_{\regB}^{f_2(i)}Z_{\regB}^{f_3(i)} \otimes I_{\regC}\right) \ket{\phi_{f_4(i)}}_{\regA \regB} \ket{i}_{\regC}$$


\noindent We present many applications below. 

\paragraph{Pseudorandomness.} We show copy expansion theorems for pseudorandom state generators and pseudorandom unitaries. 
\par Let us start with pseudorandom state generators (PRSGs). There are three versions of pseudorandom state generators that are of interest: (a) \textsc{Stretch PRSGs}: the output length of the generator, say $n$, is much larger than the key length, denoted by $\secparam$. The adversary only gets one copy of the state,  (b) \textsc{Bounded-copy PRSGs}: the number of copies received by the adversary is a priori bounded. Depending on the key length and the output length, this notion can either information theoretically exist ($t$-state designs) or require computational assumptions, (c) \textsc{Multi-copy secure PRSGs}: the number of copies received by the adversary can be an arbitrary polynomial. 
\par A number of recent works~\cite{Kretschmer21,GMMY24,LMW24,CCS24} suggest that stretch PRSGs could be strictly weaker than multi-copy PRSGs. However, the relationship between stretch PRSGs and bounded-copy PRSGs has not been thoroughly investigated so far. Using~\Cref{thm:intro:mainthm}, we show that, perhaps surprisingly, stretch PRSGs {\em do imply} bounded-copy PRSGs. As far as we are aware, this is the first copy expansion theorem for pseudorandom states. However, this implication comes at a caveat: we assume that the stretch PRSG has a bounded-sized ancilla register\footnote{In more detail, suppose the stretch pseudorandom generator $G$ can be viewed as a unitary that outputs two registers $\regA$ and $\regB$, with the pseudorandom state being on the register $\regA$ and $\regB$ is the ancilla register. We require an upper bound on the size of $\regB$ and specifically, it should be much smaller than $\regA$.}. 

\begin{theorem}[Informal]
Let $t=t(\secparam)$ be a polynomial. Assuming one-copy stretch PRSGs with some mild restrictions, there exists a $t$-copy stretch PRSGs. Specifically, we assume that the one-copy stretch PRSGs have a bounded-size ancilla register. 

If the one-copy stretch PRSG takes in keys of length $\lambda$ and outputs states over $n$ qubits, leaving some junk state on an ancilla of length $a$ qubits, then the corresponding $t$-copy stretch PRSG takes in keys of length $O(t(\lambda+a))$ and outputs states over $\geq n+a$ qubits.
\end{theorem}

In particular, there exists some constant $c$ such that if the $1$-copy PRSG maps keys of length $\secparam$ to states of length $ct \secparam$, then the corresponding $t$-copy PRSG is also expanding.

\noindent We show that by extending~\Cref{thm:intro:mainthm}, a similar copy expansion theorem can also be shown for pseudorandom unitaries (PRUs). As in the case of PRSGs, we can correspondingly define one-query, bounded-query and multi-query PRUs. We show that one-query PRUs imply bounded-query PRUs with non-adaptive security. 

\begin{theorem}[Informal]
Let $t=t(\secparam)$ be a polynomial. Assuming one-query short-key PRUs with some mild restrictions, there exists a $t$-query, non-adaptively secure PRU. Specifically, we assume that the one-query short-key PRU is "pure", that is, it clears out its ancilla register after computation.

If the one-query PRU takes in keys of length $\lambda$ and acts on states of length $n$, then the $t$-query PRU takes in keys of length $O(t\lambda)$ and acts on states of length $\geq n$.
\end{theorem}

\paragraph{Unclonable Cryptography.} We show applications of~\Cref{thm:intro:mainthm} to unclonable cryptography. Specifically, we consider two primitives: public-key quantum money~\cite{AC12,Zhandry19} and copy-protection~\cite{Aar09}. 
\par We consider a stronger security for public-key quantum money, wherein the adversary gets (unbounded) polynomially many copies of a (pure) money state associated with the same serial number. This notion was originally called quantum coins~\cite{MS09}, and prior to this work there were no known constructions. Concretely, the security guarantee states that given $t$ copies, for any polynomial $t$, of the (pure) money state, it should be computationally infeasible to produce $(t+1)$ copies of the money state. Multi-copy secure quantum money is also relevant in the setting when the quantum systems are noisy and hence, giving access to more copies would mitigate this risk. Pure multi-copy security also has applications for \textit{untraceability}, which has been studied in a recent work~\cite{CGY24a}. The property of untraceability stipulates that even the bank should not be able to trace banknotes: if every banknote is the same state then this property is immediately satisfied. 
\par We show the following. 

\begin{theorem}[Informal]
    Suppose that there exists a public key quantum money scheme and post-quantum secure pseudorandom functions exist, then there exists a multi-copy secure public-key quantum money scheme.
\end{theorem}

\begin{corollary}[Informal]
Assuming the existence of post-quantum secure indistinguishability obfuscation and post-quantum secure injective one-way functions, there exists a multi-copy secure public-key quantum money scheme. 
\end{corollary}

\noindent We similarly consider a stronger security property for quantum copy-protection as well. We require that the adversary after receiving $t$ copies of the copy-protected state is not able to produce a $(t+1)$-partite state such that all the partitions compute the original functionality. Previous works~\cite{LLLQ22,CG24} deal with the so-called {\em i.i.d-copy} security wherein the adversary receives {\em independent} copies of the copy-protected state. A couple of recent works~\cite{AMP24,PRV24} explore {\em identical-copy} security wherein the adversary receives many copies of a pure copy-protected state. However, both  works~\cite{AMP24,PRV24} considered weaker definitions of copy-protection and proposed restricted results. Using~\Cref{thm:intro:mainthm}, we show the following. 

\begin{theorem}[Informal]
For any class of functionalities ${\cal F}$, suppose there exists a copy-protection scheme for ${\cal F}$ satisfying {\em i.i.d-copy} security and that post-quantum secure pseudorandom functions exist. Then, for the same function family ${\cal F}$, there exists a copy-protection scheme satisfying {\em identical-copy} security.  
\end{theorem}

\noindent We note that i.i.d-copy security has also been explored in the context of other unclonable primitives, such as secure leasing~\cite{KNP25}. While we do not prove this formally in this work, our main theorem~\Cref{thm:intro:mainthm} also yields identical-copy security for the same primitives considered in~\cite{KNP25}.

\section{Related Work}
\label{sec:relatedwork}

There are two concurrent works showing similar results to our main theorem~\cite{CGKNY25,TWZ25}. We discuss their relationship with our results in detail.

\subsection{Comparison with concurrent work~\cite{CGKNY25}}

In particular,~\cite{CGKNY25} also shows one-copy to many-copy results in uncloneable encryption. In particular, they also construct identical-copy secure quantum money and copy-protection. As an additional result beyond ours, they show how to construct i.i.d.-copy secure uncloneable encryption, which then implies using similar techniques to ours identical-copy secure uncloneable encryption. To achieve these results, they rely on a key theorem similar to our main theorem, which they call a "purification compiler". While the description of their compiler is conceptually simpler than ours, it only works when given a purification of the form $\sum_{k} \ket{\phi_k}\ket{k}$. As such, their constructions of identical-copy secure primitives cannot be realized fully generically from any i.i.d.-copy secure scheme, and require instantiating the underlying protocols with constructions satisfying appropriate properties.

\subsection{Comparison with concurrent work~\cite{TWZ25}}

\cite{TWZ25} proves as one of its results the existence of what has later become known as a "purification channel". This is a channel (which we will suggestively name $\Sim$) with the following property. Let $\sigma$ be any mixed state of dimension $N$, and let $\mathcal{P} = \{\ketbra{\phi}\}$ be the space of all purifications of $\sigma$ of dimension $N^2$. Then $\Sim(\sigma^{\otimes t})= \E_{\ket{\phi} \mathcal{P}}[\ketbra{\phi}^{\otimes t}]$. Intuitively, this is a channel which transforms $t$ copies of a mixed state into $t$ copies of a random purification of that state. Our~\Cref{thm:intro:mainthm} can be viewed under the same lens. If we set $\mathcal{P}' = \{\ket{\phi_k}\}_{k\in \mathcal{K}}$, then we give a channel $\Sim$ such that $\Sim(\sigma^{\otimes t}) \approx \E_{\ket{\phi_k} \gets \mathcal{P}'}[\ketbra{\phi_k}^{\otimes t}]$.

Our result thus has a few main differences from~\cite{TWZ25}. First, the distribution over resulting purifications is different. The~\cite{TWZ25} result uses a truly random purification, whereas our result uses a random element from some particular ensemble of purifications. Note that using the particular ensemble of purifications we consider has the advantage that if you can simulate a random function our ensemble is efficient to sample from. Our construction also has slightly more error than~\cite{TWZ25}. On the other hand, our construction is significantly easier to understand and analyze, as we do not require any techniques from representation theory. Our result has the additional advantage that the dimension of the purification ensemble can be scaled arbitrarily, although the error of our construction is inversely correlated with the dimension of the purification. In terms of applications, our result can be used to recover all applications used in~\cite{TWZ25}. Although it is not immediate, it seems likely that the reverse is also true. Note that for the application of purification channels to tomography~\cite{PSTW25}, our construction is sufficient but gives significantly worse bounds.

\subsection{Subsequent work}

\cite{GMMY24} originally showed a number of black-box separations between cryptographic primitives relative to a CPTP oracle. A recent update to this paper used a variation of our main theorem to lift these result to unitary oracle separation. 

\cite{YNM25} demonstrates how to simulate parallel queries to a random dilation of an arbitrary quantum channel using as many parallel queries to the channel itself. Viewed in this way, our~\Cref{thm:unisimmain,thm:isosim} shows the same result for the specific channel which applies samples a random key and applies the map $U_k$ for some family of unitaries or isometries $\{U_k\}$. In particular, our construction of many-query non-adaptive pseudorandom unitaries from one-query secure pseudorandom unitaries can be thought of as applying a randomized dilation of the channel: pick a random unitary from the family and apply it (see~\Cref{fig:prucons}).
\section{Technical Overview}
\label{sec:techoverview}

\pnote{I'm not sure if compressed oracle is the best way to lead in to our paper. This is how I would do it (In retrospect, this is probably not how we came up with the idea but our idea looks more natural if we think of it this way): I would start with the binary phase PRS (or even the one with the large domain), explain the proof using the purification technique.}

\paragraph{Background: the compressed oracle method} Quantum queries to a random function are most commonly analyzed using the compressed oracle framework~\cite{zhandry2019record}. We will model a random function $f$ generally as a phase oracle $S^f$, defined by the map
$$S^f\ket{x}\mapsto (-1)^{f(x)}\ket{x}.$$
In the compressed oracle framework, the mixed state resulting from some algorithm $\A^f$ querying a random function $S^f$ is instead modelled by its purification. Instead of representing $\ket{\A^f}$ for a random $f$ as the mixed state
$$\E_f[\ketbra{\A^f}]$$
the compressed oracle framework will consider the pure state
$$\sum_f \ket{\A^f}\ket{f}$$
Tracing out the $f$ register results in the original mixed state.

The key insight behind the compressed oracle framework is that taking the quantum Fourier transform of the $\ket{f}$ register leaves behind a transcript recording the queries made to $f$ by $\A$. Since this transcript will always contain at most as many queries as made by the algorithm $\A$, it can actually be represented efficiently. In particular, let $\regD$ be a register containing a set, initialized to $\emptyset$. Define the map 
$$CO\ket{x}\ket{D}_{\regD} \mapsto \begin{cases}
    \ket{x}\ket{D\setminus \{x\}}_{\regD}&x\in D\\
    \ket{x}\ket{D\cup \{x\}}_{\regD}&x\notin D
\end{cases}$$
It turns out that oracle access to $CO$ is equivalent to oracle access to $S^f$. That is, for any algorithm $\A$, $\E_f[\A^{f}]$ represents the same density matrix as $\Tr_{\regD}(\A^{CO})$.

Note that in the compressed oracle model, when $CO$ is queried twice on the same $x$ it will erase its saved state. While it turns out that this behavior can be very useful, for our purposes we would prefer that the compressed oracle actually tracks all queries made to $x$. It turns out that a slight generalization of the compressed oracle model to multi-bit random functions can easily achieve this goal. In particular, let $f:\{0,1\}^n\to [t]$ be a random function outputting a value in $[t]$. We will redefine $S^f$ to be the following map
$$S^f\ket{x}\mapsto \omega_t^{f(x)}\ket{x}$$
where $\omega_t$ is a $t$-th root of unity (so $\omega_t^t=\omega_t^0=1$).

We define the following expanded compressed oracle, where $D$ will now represent a multiset:
\begin{equation}\label{eq:co1}
    CO_t:\ket{x}\ket{D}_{\regD}\mapsto \ket{x}\ket{D\uplus \{x\}}_{\regD}
\end{equation}
Now, as long as an algorithm $\A^{(\cdot)}$ makes $<t$ queries, it is again the case that 
\begin{equation}\label{eq:co2}
    \E_{f:\{0,1\}^n\to [t]}[\A^f] = \Tr_{\regD}(\A^{CO_t})
\end{equation}
Note that setting $t=2^n$ allows us to handle all efficient algorithms $\A^{(\cdot)}$.

\paragraph{The main challenge} All of our results center around a solution for the following challenge: for a mixed state $\rho$, is it possible to construct a family of pure states $\ket{\psi_k}$ such that many copies of $\ket{\psi_k}$ function like many copies of $\rho$? 

An ideal solution to this question should look something like the following. Using $\ket{\psi_k}^{\otimes t}$, it should be possible to generate $\rho^{\otimes t}$. Similarly, using $\rho^{\otimes t}$, it should be possible to generate the mixed state $$\E_k\left[\ketbra{\psi_k}^{\otimes t}\right]$$

\paragraph{Example: random states from a family}
Let's start by considering a common example. Let $\{\ket{\phi_k}\}$ be some family of states. Let us consider $\rho$ the following distribution:
\begin{enumerate}
    \item Sample $i$ at random.
    \item Output $\ket{\phi_i}\ket{i}\ket{i}$
\end{enumerate}
As a mixed state, $\rho = \E_i [\ketbra{\phi_i}\otimes \ketbra{i}]$.

Our construction will be simple. Let $f:\{0,1\}^n\to \{0,1\}^n$ be some function sampled at random. Our state $\ket{\psi_f}$ will be defined by 
$$\ket{\psi_f} = \sum_{i} \omega_{2^n}^{f(i)} \ket{\phi_i}\ket{i}$$

This construction works because of the following key idea: \textbf{applying a random phase to a state is essentially the same as measuring it in the standard basis.}

In particular, consider generating $\ket{\psi}$ using $CO_{2^n}$ instead of $f$. Applying~\Cref{eq:co1,eq:co2} gives us
$$\E_{f}[\ketbra{\psi}] \propto \Tr_{\regD}\left(\sum_i \ket{\phi_i}\ket{i}\ket{\{i\}}_\regD\right)$$

Generalizing to $t$ copies we get
$$\E_{f}[\ketbra{\psi}^{\otimes t}] \propto \Tr_{\regD}\left(\sum_{i_1,\dots,i_t} \left(\bigotimes_{j=1}^t\ket{\phi_{i_j}}\ket{i_j}\right)\ket{\{i_1,\dots,i_t\}}_\regD\right)$$

Measuring the $\regD$ register tracks exactly what values $f$ was applied to, and so the residual state will be as if all $i$'s were measured, with the order information forgotten. Intuitively, applying a random phase oracle to $t$ different states measures all of them, but also permutes the order. Formally,
\begin{equation}
\begin{split}
\Tr_{\regD}\left(\sum_{i_1,\dots,i_t} \left(\bigotimes_{j=1}^t\ket{\phi_{i_j}}\ket{i_j}\right)\ket{\{i_1,\dots,i_t\}}_\regD\right)\\
= \Tr_{\regD}\left(\sum_{i_1,\dots,i_t}\sum_{\pi\in Sym(t)} \left(\bigotimes_{j=1}^t\ket{\phi_{i_{\pi(j)}}}\ket{i_{\pi(j)}}\right)\ket{\{i_1,\dots,i_t\}}_{\regD}\right)\\
= \E_{i_1,\dots,i_t}\left[\left(\sum_{\pi\in Sym(t)} \left(\bigotimes_{j=1}^t\ket{\phi_{i_{\pi(j)}}}\ket{i_{\pi(j)}}\right)\right)\left(\sum_{\pi\in Sym(t)} \left(\bigotimes_{j=1}^t\bra{\phi_{i_{\pi(j)}}}\bra{i_{\pi(j)}}\right)\right)\right]
\end{split}
\end{equation}

We can generate exactly the state $\E_f[\ketbra{\psi}^{\otimes t}]$ by sampling $i_1,\dots,i_t$ uniformly at random and then generating the state 
$$\sum_{\pi \in Sym(t)}\left(\bigotimes_{j=1}^t\ket{\phi_{i_{\pi(j)}}}\ket{i_{\pi(j)}}\right)$$

\paragraph{Erasing the index} Note that for most applications (such as pseudorandom states), the distribution we care about is 
$$\E_{i}\left[\ketbra{\phi_i}\right]$$
where the index $i$ is not revealed. Resolving this is simple, we simply hide the index behind another random function. In particular, let $f_1,f_2$ be two random functions. If we define
$$\ket{\psi_{f_1,f_2}} \propto \sum_{i} \omega_{2^n}^{f_1(i)} \ket{\phi_{f_2(i)}}\ket{i}$$
we can then generate exactly the state $\E_{f_1,f_2}\left[\ketbra{\psi_{f_1,f_2}}^{\otimes t}\right]$ by sampling $i_1,\dots,i_t$, $r_1,\dots,r_t$ uniformly at random and then generating the state
$$\sum_{\pi \in Sym(t)}\left(\bigotimes_{j=1}^t\ket{\phi_{r_{\pi(j)}}}\ket{i_{\pi(j)}}\right)$$

\paragraph{Handling general mixed states} To generalize this to mixed states, we make the observation that, by appending randomness, every mixed state looks like a random pure state from some family. In particular, let $\ket{\phi_k}_{\regA \regB}$ be the purification of some mixed state $\rho_{k,\regA}=\Tr_{\regB}(\ketbra{\phi_k}_{\regA \regB})$. Then applying a quantum one-time pad to the $\regB$ register exactly looks like tracing out $\regB$ and appending randomness. That is,
$$\Tr_{\regB}(\ketbra{\phi_k}_{\regA \regB})=\E_{x,z}\left[(I_{\regA}\otimes X^x_{\regB}Z^z_{\regB})\ketbra{\phi_k}(I_A\otimes Z^z_{\regB}X^x_{\regB})\right]$$
In particular, if we define the family $\ket{\phi_{k,x,z}'} = (I_{\regA}\otimes X^xZ^z)\ket{\phi_k}$, then
\begin{equation}\label{eq:otp}
    \rho_{\regA} \otimes I_{\regB}=\Tr_{\regB}(\ketbra{\phi}_{\regA \regB})\otimes I_{\regB} = \E_{x,z}[\ketbra{\phi_{x,z}'}_{\regA}]
\end{equation}
Thus, if we define 
$$\ket{\psi_{f_1,f_2}}\propto \sum_{i}\omega_{2^n}^{f_1(i)}\ket{\phi_{f_2(i)}'}\ket{i}$$
this looks like taking a few samples of $\rho$ and permuting them. 

In full detail, for a family $\rho_k = \Tr_B(\ket{\phi_k})$, we will define
$$\ket{\psi_{f_1,f_2,f_3,f_4}}\propto \sum_i \omega_{2^n}^{f_1(i)} (I_{\regA} \otimes X^{f_2(i)}_{\regB} Z^{f_3(i)}_{\regB})\ket{\phi_{f_4(i)}}\ket{i}$$

Now (with some error coming from the probability that $i_j=i_{j'}$), we have
$$\E_{f_1,f_2,f_3,f_4}[\ketbra{\psi_{f_1,f_2,f_3,f_4}}^{\otimes t}] \approx \Sim(\rho^{\otimes t})$$
where $\Sim(\rho^1,\dots,\rho^t)$ is defined by
\begin{enumerate}
    \item Sample $i_1,\dots,i_t$ at random
    \item Sample $r_1,\dots,r_t$ at random
    \item Generate the state
    $$\sum_{\pi\in Sym(t)}\bigotimes_{j=1}^t \ketbra{0}_{A_j} \otimes \ketbra{i_{\pi(j)}}_{\regK_j} \otimes \ketbra{r_{\pi(j)}}_{\regB_j}$$
    \item Swap $\rho^{i_j}$ into $\regA_j$ controlled on $K_j$ containing $k_i$.
    \item Output registers $\regA_1 \regK_1,\dots,\regA_t \regK_t$.
\end{enumerate}

And so, we can simulate a random $\ket{\psi_{f_1,f_2,f_3,f_4}}$ using samples from $\rho_k$. But the opposite is also true. Tracing out the $\ket{i}$ and $\regB$ registers in $\ket{\psi_{f_1,f_2,f_3,f_4}}$ leaves us exactly with the state $\rho$. And so for most purposes, $\ket{\psi_{f_1,f_2,f_3,f_4}}^{\otimes t}$ acts like $\rho^{\otimes t}$.

\subsection{Applications}

\paragraph{Uncloneable cryptography}
This result immediately implies that any cryptographic primitive with "i.i.d." security can be converted into one with "pure" security using a pseudorandom function. In particular, if a cryptographic protocol satisfies security against adversaries given "many copies" of some mixed output state, we can replace that mixed state with the pure state described in~\Cref{thm:intro:mainthm}. 

As applications of this idea, we show how to construct identical-copy secure copy protection and quantum money from i.i.d.-copy secure versions of both primitives. 

As an example, we will detail the argument for public key quantum money. An (i.i.d. secure) public key quantum money protocol is a triplet $(\gen,\mint,\verify)$, where
\begin{enumerate}
    \item $\gen \to (\pk,\sk)$ outputs a public key, secret key pair
    \item $\mint(\sk)\to (\rho,s)$ takes in a secret key, and produces a money state $\rho$ along with a serial number $s$
    \item $\verify(\pk, \rho, s) \to \{0,1\}$ checks whether the input $\rho,s$ is valid for the public key $\pk$
\end{enumerate}
The security definition stipulates that for any polynomial $t$, given $\mint(\sk)^{\otimes t} \to (\rho_1,s_1,\dots,\rho_t,s_t)$ it should be computationally infeasible to find $t+1$ money state, serial number pairs which pass verification (even if the money states can be entangled). 

But note that for any given $\sk$, $\mint(\sk)=(\rho,s)$ is some mixed state. And so we will set $\{\ket{\psi_k^{\sk}}\}$ to be the purification ensemble such that for a random $k$, $\ket{\psi_k}^{\otimes t}$ can be simulated using $\mint(\sk)^{\otimes t}$. 

Then, we can see that for a random $k$, $\ket{\psi_k}$ is $t\to t+1$ uncloneable. Given $\ket{\psi_k}^{\otimes t}$, it should be impossible to generate $\ket{\psi_k}^{\otimes t+1}$. If we could, then we could run the same procedure on $\Sim(\mint(\sk)^{\otimes t})$ to generate $t+1$ accepting money states. In particular, the sampler which outputs $(\ket{\psi_k^{\sk}},\pk)$ will be what we call a "multi-copy secure mini-scheme", which using standard techniques can be upgraded into a multi-copy secure quantum money protocol.

Similar (albeit simpler) approaches work generically for most primitives with cryptographic uncloneability satisfying "i.i.d.-security", and in the main body of the paper we show that this works for a broad class of such objects.

\paragraph{Pseudorandom states}
We can also use our main theorem (\Cref{thm:intro:mainthm}) to construct $t$-copy secure pseudorandom states from one-copy secure pseudorandom states. We will assume without loss of generality that the one-copy secure pseudorandom state generator $G(k)$ acts as follows
\begin{enumerate}
    \item Apply a unitary $U_G$ to the state $\ket{k}\ket{0}$, producing a state $\ket{\phi_k}_{\regA \regB}$.
    \item Output $\Tr_{\regB}(\ketbra{\phi}_{\regA\regB})$.
\end{enumerate}
Then, we can instantiate~\Cref{thm:intro:mainthm} with $\{\ket{\phi_k}_{\regA \regB}\}$, where $f_1,f_2,f_3,f_4$ are $2t$-wise independent hash functions. We get a family of states $\{\ket{\psi_{\wt{k}}}\}$ such that 
$$\E_{\wt{k}}\left[\ketbra{\psi_{\wt{k}}}^{\otimes t}\right]$$
can be simulated with $G(k_1),\dots,G(k_t)$ for $k_1,\dots,k_t$ chosen at random.

By one-copy security, we can replace each $G(k_1),\dots,G(k_t)$ with a Haar random state. Since one copy of a Haar random state is indistinguishable from a random string, we get 
$$\E_{\wt{k}}\left[\ketbra{\psi_{\wt{k}}}^{\otimes t}\right] \approx \E_{r_1,\dots,r_t}\left[\Sim(\ketbra{r_1}\otimes \dots \otimes \ketbra{r_t})\right]$$
where $\Sim$ is the algorithm from~\Cref{thm:intro:mainthm}.

We complete the argument by explicitly computing the mixed state
$$\E_{r_1,\dots,r_t}\left[\Sim(\ketbra{r_1}\otimes \dots \otimes \ketbra{r_t})\right]$$
and showing that it is statistically close to $t$ copies of a Haar random state. And so, $\ket{\psi_{\wt{k}}}^{\otimes t}$ for a random key $\wt{k}$ is also indistinguishable from a Haar random state.

Note that the key $\wt{k}$ contains a key for a $2t$-wise independent hash function with output length the length of the ancilla register. And so, our key grows with the number of ancillas used by the construction.

\subsection{Extension to unitaries}

In order to achieve copy-expansion for pseudorandom unitaries, we first prove a variant of~\Cref{thm:intro:mainthm} for the unitary setting. 
\begin{theorem}[Unitary Setting Main Theorem; Informal]
    Consider a family of unitaries $\{U_i\}_{i \in \{0,1\}^n}$. Then there exists another family of unitaries $\{\wt{U}_{\wt{k}}\}$ such that one parallel query to $\wt{U}_{\wt{k}}$ for a random $\wt{k}$ (formally, a single query to the map $\E_{\wt{k}}[\wt{U}_{\wt{k}}^{\otimes t}]$) can be efficiently approximated by making a single query to the map
    $$\E_{r_1,\dots,r_t}\left[U_{r_1}\otimes \dots \otimes U_{r_t}\right]$$
\end{theorem}

The proof then follows roughly the same structure as copy-expansion for pseudorandom states. In particular, for $\{U_k\}$ a (pure) pseudorandom unitary family, a parallel query to $t$-copies of a unitary from $\{\wt{U}_{\wt{k}}\}$ will be indistinguishable from a single query to the simulator, which queries $U_{r_1}\otimes \dots \otimes U_{r_t}$ once for random $r_1,\dots,r_t$. Since $\{U_k\}$ is a one-time pseudorandom unitary, we can thus replace $U_{r_1},\dots, U_{r_t}$ with truly random unitaries. It then remains to show that the simulator we define when instantiated with truly random unitaries is itself indistinguishable from $t$ queries to a truly random unitary. We prove this via a careful use of the path-recording method from~\cite{MH24}.

Note that we need the pseudorandom unitary to be pure in order to make implementing our construction possible in the first place.

We thus get that if one-copy secure pseudorandom unitaries with sufficiently compact keys exist, then $t$-copy non-adaptively secure pseudorandom unitaries exist with a key that grows linearly with the number of copies.\\
\section{Preliminaries}
\label{sec:prelims}
\noindent We denote the security parameter to be $\secparam$. We denote $\negl(\cdot)$ to be a negligible function. We denote ${\sf wt}(\cdot)$ to be the Hamming weight. 

\paragraph{Notation.} A register $\regA$ is a named finite-dimensional Hilbert space. If $\regA$ and $\regB$ are registers, then $\regA \otimes \regB$ denotes the tensor product of the two associated Hilbert spaces. For a set $S$, we denote $\mathcal{H}(S)$ to be the $|S|$-dimensional Hilbert space spanned by $\ket{x}$ for $x\in S$. We define $\Unitary(\hilbert_\regA)$ as the set of all unitary operators acting on a Hilbert space $\Unitary(\hilbert_\regA)$. For $N \in \N$, we define $\Unitary(N)$ as the set of all $N$-dimensional unitary operators. We denote the identity operator on a register $\regA$ to be $I_{\regA}$. On an $(m+n)$-qubit state $\ket{\psi}$, if an $m$-qubit unitary $U$ is applied on the first $m$ qubits and a unitary $V$ is applied on the last $n$ qubits then we denote this by $(U_m \otimes V_n)\ket{\psi}$.
\par We denote by $\TD(\rho, \rho') = \frac{1}{2} \norm{\rho - \rho'}_1$ the trace distance between operators $\rho$ and $\rho'$, where $\norm{X}_1 = \Tr(\sqrt{X^{\dagger} X})$ is the trace norm. We use $\rho \approx_{\varepsilon} \rho'$ to denote the fact that $\TD(\rho,\rho')=\varepsilon$.
\par We denote $Sym([t])$ to be the symmetric group, consisting of all the permutations mapping $[t]$ to $[t]$.

\subsection{Type States}
\begin{definition}[Type vectors]
    Denote $[s]_0 \coloneqq [s] \cup \{0\}$. An $(\ell,s)$-type vector is a vector $T \in [s]_0^{\ell}$ such that $\sum_i T_i = s$. 

    We denote 
    $$\types(\ell, s) \coloneqq \left\{T\in [s]_0^\ell : \sum_i T_i = s\right\}$$
    to be the set of $(\ell,s)$-type vectors.

    For any $s$ length vector $\vec{v} \in [\ell]^s$, we say the type of $\vec{v}$ matches $T$, or $\type(\vec{v}) = T$ if, for each $i \in [\ell]$, the number of times $i$ appears in $\vec{v}$ is exactly $T_i$.

    Each type vector then defines a set of "matching" vectors in $[\ell]^s$. We define
    $$S_T\coloneqq \{\vec{v}\in [\ell]^s : \type(\vec{v}) = T\}.$$
    For ease of notation, if $\type(\vec{v}) = T$, we will sometimes write $\vec{v} \in T$ instead of $\vec{v} \in S_T$. Similarly, we will denote $|T|\coloneqq |S_T|$.
\end{definition}

\begin{definition}
\label{def:prelims:typestates}
Let $\ell=\poly(\secparam)$ and $t \leq \ell$. Suppose $\{\ket{\psi_i}\}_{i \in [\ell]}$ be an arbitrary collection of $n$-qubit states.  We define the following state $\ket{\setst_{\Psi,t,u}}$, for any $u \in [t]^{2^n}$ as follows: 
$$\ket{\setst_{\Psi,t,u}} \propto \sum_{\type((i_1,\ldots,i_t))= u} \ket{i_1 \ldots i_t}\ket{\psi_{i_1} \cdots \psi_{i_t}},$$
where $\Psi=\{\ket{\psi_{i_1}},\ldots,\ket{\psi_{i_t}}\}$.
\end{definition}

\subsection{$t$-wise independence and $t$-designs}


\begin{definition}
    A $t$-wise independent hash function is a family of functions $\{f_k:[N]\to [M]\}$ such that for all $x_1\neq \dots \neq x_t\in [N]$, for all $y_1,\dots,y_t\in [M]$,
    $$\Pr_{k\gets [N]}[f_k(x_1)=y_1\wedge\dots\wedge f_k(x_t)=y_t] = \frac{1}{M^t}$$
\end{definition}

\begin{theorem}[\cite{2twise}]\label{lem:twise}
    Let $\mathcal{F}=\{f:[N]\to [M]\}$. Let $\{f_k:[N]\to[M]\}$ be a $2t$-wise independent hash function family. Let $\A^{(\cdot)}$ be any (possibly inefficient) $t$ query quantum algorithm. Then
    $$\abs{\Pr_{k\gets [N]}[\A^{f_k}\to 1]-\Pr_{f\gets \mathcal{F}}[\A^f\to 1]} = 0$$
\end{theorem}

\begin{theorem}[\cite{efftwise}]\label{thm:efftwise}
    For all $t,n,\ell'<n$, there exists a $t$-wise independent hash function from $\ell'$ bits to $n$ bits with key length $O(tn)$.
\end{theorem}

\begin{definition}
    A $\epsilon$-approximate unitary $t$-design is a family of unitaries $\{U_k\}$ over $\mathcal{H}([N])$ such that for all (possibly inefficient) $t$ non-adaptive query quantum algorithms $\A^{(\cdot)}$:
    $$\abs{\Pr_{k\gets [N]}[\A^{\U_k}\to 1] - \Pr_{U\gets \Haar{[N]}}[\A^{U}\to 1]} \leq \epsilon$$
\end{definition}

\begin{theorem}[\cite{designs}]\label{thm:effdesign}
    For all $n\in \N,\epsilon\in (0,1)$, there exists an $\epsilon$-efficient approximate $t$-design on $\mathcal{H}(\{0,1\}^n)$ with key length $O\left(nt+\log\frac{1}{\epsilon}\right)$.
\end{theorem}

\subsection{Pseudorandomness}

\subsubsection{Haar Measure}
\label{sec:haarmeasure}
\noindent Haar measure is a unique left-invariant (and right-invariant) measure on the unitary group. We denote the Haar measure on the $n$-qubit unitary group to be $\Haar_n$. 
\par We state a well known fact on 1-designs below. 

\begin{theorem}
Suppose $\rho_{\regA \regB}$ is an $n+m$ qubit state on two registers $\regA$ (first $n$ qubits) and $\regB$ (last $m$ qubits). Then the following holds: 
$$\E_{\substack{a \xleftarrow{\$} \{0,1\}^n\\ b \xleftarrow{\$} \{0,1\}^n}} \left[ \left(X^{a}_{\regA}Z^b_{\regA} \otimes I_{\regB} \right) \rho \left(X^{a}_{\regA}Z^b_{\regA} \otimes I_{\regB} \right) \right] = \frac{1}{2^n} I_{\regA} \otimes \Tr_{\regA}\left( \rho_{\regA \regB} \right) $$
\end{theorem}

\noindent Using the Haar measure, we can correspondingly define a distribution on quantum states. We define the Haar distribution on $n$-qubit states, denoted by $\mu_n$, to be the following distribution: output $U\ket{0^n}$, where $U$ is sampled from the Haar measure $\Haar_n$.

\subsubsection{Pseudorandom States}
\label{sec:defs:prs} 
\noindent We recall the notion of pseudorandom state generators, a computational generalization of $t$-state designs.

\begin{definition}[Pseudorandom State Generator~\cite{JLS18}]
\label{def:prsg}
Let $\ell_k(\secparam),\ell_n(\secparam)$ be polynomially bounded functions. A $(\ell_k,\ell_n)$-\emph{pseudorandom state generator} (PRSG) is a polynomial-sized quantum algorithm $G$ that takes as input a classical string $k \in \{0,1\}^{\ell_k}$ (called the \emph{seed}) and outputs an $\ell_n$-qubit quantum state $G(k)=\rho_k$. It satisfies the following property: For any polynomial $t = t(\lambda)$ and any quantum polynomial-time algorithm $\adversary$ that receives $t$ copies of either $G(k)$ or a Haar random $\ell_n$-qubit state $\ket{\psi}$, we have:
\begin{align}
\left| \Pr_{k \leftarrow \{0,1\}^{\ell_k(\secparam)}} [\adversary(1^{\secparam},G(k)^{\otimes t}) = 1] - \Pr_{\ket{\psi} \leftarrow \mu_{\ell_n}} [\adversary(1^{\secparam},\ket{\psi}^{\otimes t}) = 1] \right| \leq \negl(\lambda)
\end{align}
where the probability is taken over the choice of the seed $k$, the Haar random state $\ket{\psi}$, and the internal randomness of $\adversary$.
\end{definition}

\noindent \cite{JLS18,BrakerskiS20,AGQY22} showed that pseudorandom state generators exist under the assumption of post-quantum assumptions. 
\par We consider a bounded copy variant of the above definition below. 

\begin{definition}[Bounded-Copy Pseudorandom State Generator~\cite{JLS18}]
\label{def:bounded:prsg}
Let $\ell_k,\ell_n$ be two polynomially bounded function and let $t = t(\lambda)$ be a polynomial. A $(\ell_k,\ell_n,t)$-\emph{pseudorandom state generator} (PRSG) is a polynomial-sized quantum algorithm $G$ that takes as input a classical string $k \in \{0,1\}^{\ell_k}$ (called the \emph{seed}) and a parameter $t$, and outputs an $\ell_n$-qubit quantum state $G(k,t)=\ket{\psi_k}$. It satisfies the following property: For any quantum polynomial-time algorithm $\adversary$, we have:
\begin{align}
\left| \Pr_{k \leftarrow \{0,1\}^{\ell_n}} [\adversary(1^{\secparam},G(k)^{\otimes t}) = 1] - \Pr_{\ket{\psi} \leftarrow \mu_{\ell_n}} [\adversary(1^{\secparam},\ket{\psi}^{\otimes t}) = 1] \right| \leq \negl(\lambda)
\end{align}
where the probability is taken over the choice of the seed $k$, the Haar random state $\ket{\psi}$, and the internal randomness of $\adversary$.
\end{definition}

\noindent Note that if $\ell_k \geq \ell_n\cdot t$ then it is possible to even achieve statistical security (i.e. $t$-state designs) in the above definition. However, we do not place any such restrictions in this work. In more detail, our results will hold for inherently cryptographic parameter regimes, e.g. $\ell_k \ll \ell_n$. \pnote{define stretch PRSGs here}

\subsubsection{Pseudorandom Unitaries}
\noindent In the same work,~\cite{JLS18} also defined pseudorandom unitaries, which are a computational generalization of $t$-unitary desgins. As in the case of~\Cref{sec:defs:prs}, we present two definitions of pseudorandom unitaries. In the first definition, the adversary can make a priori unboudned number of oracle queries whereas in the second definition, it is restricted to only make a bounded number of queries.

\begin{definition}[Pseudorandom Unitary]\label{def:pru}
    Let $\ell_k(\secparam),\ell_n(\secparam),\ell_a(\secparam)$ be polynomials. An $(\ell_k,\ell_n)$-pseudorandom unitary is an efficient family of unitaries $\{\PRU_\secparam\}_{\secparam\in \N}$ defined on registers $\regIn$ over $\{0,1\}^{\ell_n(\secparam)}$, $\regK$ over $\{0,1\}^{\ell_k(\secparam)}$, $\regAnc$ over $\{0,1\}^{\ell_a(\secparam)}$ for polynomials satisfying the following property:
    
    For $k\in \{0,1\}^{\secparam}$, let $\PRU_k(\cdot)$ be the CPTP map which on input $\rho_{In}$ outputs
    $$\PRU_k(\rho)\coloneqq \Tr_{\regK,\regAnc}\left(\PRU_{\secparam} (\ketbra{k}_{\regK} \otimes \rho_{\regIn}\otimes \ketbra{0}_{Anc}) \PRU_{\secparam}^{\dagger}\right)$$

    We say that a pseudorandom unitary is secure if for all non-uniform QPT oracle adversaries $\A$,
    $$\abs{\Pr_{k\gets \{0,1\}^{\ell_k(\secparam)}}\left[\A^{\PRU_k}(1^\secparam)\to 1\right] - \Pr_{U\gets \Haar(\{0,1\}^{\ell_n(\secparam)})}\left[\A^{U}(1^\secparam)\to 1\right]} \leq \negl(\secparam)$$
\end{definition}

\noindent~\cite{MPSY24,MH24} showed that pseudorandom unitaries exist under the assumption of post-quantum one-way functions.

\begin{definition}[(Non-adaptive) Bounded-Query Pseudorandom Unitary]\label{def:prubounded}
    Let $\ell_k,\ell_n,\ell_a,t$ be polynomials. A (non-adaptive) $(\ell_k,\ell_n,t)$-pseudorandom unitary is an efficient family of unitaries $\{\PRU_\secparam\}_{\secparam\in \N}$ defined on registers $\regIn$ over $\{0,1\}^{\ell_n(\secparam)}$, $\regK$ over $\{0,1\}^{\ell_k(\secparam)}$, $\regAnc$ over $\{0,1\}^{\ell_a(\secparam)}$ for polynomials satisfying the following property:
    
    For $k\in \{0,1\}^{\secparam}$, let $\PRU_k(\cdot)$ be the CPTP map which on input $\rho_{In}$ outputs
    $$\PRU_k(\rho)\coloneqq \Tr_{\regK,\regAnc}\left(\PRU_{\secparam} (\ketbra{k}_{\regK} \otimes \rho_{In}\otimes \ketbra{0}_{Anc}) \PRU_{\secparam}^{\dagger}\right)$$

    We say that a pseudorandom unitary is (non-adaptive) $t$-copy secure if for all (non-adaptive) non-uniform QPT oracle adversaries $\A$ making at most $t$ queries,
    $$\abs{\Pr_{k\gets \{0,1\}^{\ell_k(\secparam)}}\left[\A^{\PRU_k}(1^\secparam)\to 1\right] - \Pr_{U\gets \Haar(\{0,1\}^{\ell_n(\secparam)})}\left[\A^{U}(1^\secparam)\to 1\right]} \leq \negl(\secparam)$$
\end{definition}

\ifnum\llncs=0
\subsection{Uncloneable Cryptography}
\ifnum\llncs=1
\section{Preliminaries: Quantum Money and Copy-protection} 
\fi

\subsubsection{Quantum Money}
\noindent We first recall the definition of a quantum money mini scheme~\cite{AC12}. In this notion, there is a minting algorithm that produces a publicly verifiable quantum money state along with a serial number. Moreover, in terms of security, we require that the quantum money state cannot be cloned. 

\begin{definition}[Quantum Money Mini-Scheme]
\label{def:mini-scheme}
A \emph{quantum money mini-scheme} is a pair of QPT algorithms $(\mathsf{Mint}, \mathsf{Ver})$ where:
\begin{itemize}
\item $\mathsf{Mint}(1^{\secparam}) \to (\rho, s)$: It takes a security parameter $1^{\secparam}$ and outputs a quantum state $\rho$ and a classical serial number $s$.

\item $\mathsf{Ver}(s, \sigma) \to \{0,1\}$: It takes a serial number $s$ and a quantum state $\sigma$, and outputs 1 (accept) or 0 (reject).
\end{itemize}

The scheme must satisfy:

\begin{enumerate}
\item \textbf{Correctness:} For all $(\rho, s) \leftarrow \mathsf{Mint}(1^{\secparam})$:
\begin{align}
\Pr[1 \leftarrow \mathsf{Ver}(s, \rho)] = 1
\end{align}

\item \textbf{Security:} For any polynomial-time quantum adversary $\mathcal{A}$:
\begin{align}
\Pr\left[\begin{array}{c}
(\rho, s) \leftarrow \mathsf{Mint}(1^{\secparam}) \\
(\sigma_1, \sigma_2) \leftarrow \mathcal{A}(\rho, s) \\
1 \leftarrow \mathsf{Ver}(s, \sigma_1) \land 1 \leftarrow \mathsf{Ver}(s, \sigma_2) 
\end{array}\right] \leq \negl(n)
\end{align}
\end{enumerate}
\end{definition}

\noindent We now recall the definition of public-key quantum money. The main difference between the mini scheme and the definition below is that in the mini scheme, anyone can produce a money state whereas in the definition below, only the one who possesses the secret key can produce the state. Furthermore, public-key quantum money satisfies $t\to t+1$ copy security, in the sense that given $t$ money states no adversary can produce $t+1$ money states. However, since money states for a public-key quantum money protocol are allowed to be mixed, one should think of public-key quantum money as an "i.i.d." version of a mini scheme. Using digital signatures, any mini scheme can be upgraded into a public-key quantum money scheme. 

\begin{definition}[Public-Key Quantum Money]
\label{def:pkqm}
A \emph{public-key quantum money scheme} is a tuple of QPT algorithms $(\mathsf{Gen},\mathsf{Mint},\mathsf{Ver})$ where:
\begin{itemize}
\item $\mathsf{Gen}(1^{\secparam}) \to (\pk,\sk)$: It takes a security parameter $1^{\secparam}$ and outputs a public key $\pk$ and a secret key $\sk$.

\item $\mathsf{Mint}(\sk) \to (\rho, s)$: It takes a secret key $\sk$ and outputs a quantum state $\rho$ and a classical serial number $s$.

\item $\mathsf{Ver}(\pk, \sigma, s) \to \{0,1\}$: It takes as input the verification key $\pk$, a serial number $s$ and a quantum state $\sigma$, and outputs 1 (accept) or 0 (reject).
\end{itemize}

The scheme must satisfy:

\begin{enumerate}
\item \textbf{Correctness:} For all $(\pk,\sk) \leftarrow \mathsf{Gen}(1^{\secparam})$, $(\rho, s) \leftarrow \mathsf{Mint}(\sk)$:
\begin{align}
\Pr[1 \leftarrow \mathsf{Ver}(\pk,s, \rho)] = 1
\end{align}

\item \textbf{i.i.d. Security:} For any polynomial $t$ and any polynomial-time quantum adversary $\mathcal{A}$:
\begin{align}
\Pr\left[\begin{array}{c}
(\pk,\sk) \leftarrow \mathsf{Gen}(1^{\secparam}) \\
(\rho_1,s_1,\dots,\rho_t,s_t) \leftarrow \mathsf{Mint}(\sk)^{\otimes t} \\
(s_1,\sigma_1,\dots,s_{t+1},\sigma_{t+1}) \leftarrow \mathcal{A}(\rho_1,s_1,\dots,\rho_t,s_t) \\
\text{For all }i\in [t+1], 1\gets \verify(\pk, \sigma_i, s_i)
\end{array}\right] \leq \negl(\secparam)
\end{align}
where $\sigma_1,\dots,\sigma_{t+1}$ may possibly be entangled.
\end{enumerate}
\end{definition}

\paragraph{Multi-copy Security.} We consider a strengthing of the above definitions wherein the adversary can receive many copies of the quantum state. We present these notions below. 

\begin{definition}[Multi-copy secure mini scheme]
\label{def:multicopy:mini-scheme}
A \emph{multi-copy secure quantum money mini scheme}, defined by a pair of QPT algorithms $({\sf Mint},{\sf Ver})$ is a mini scheme and additionally, it satisfies the following properties:

\begin{enumerate}

\item \textbf{Purity}: the output of ${\sf Mint}$ is a pure state. Concretely, ${\sf Mint}$ proceeds in the following steps:
\begin{itemize}
    \item It first generates $s$ along with secret randomness $sk$. 
    \item Apply an isometry $U$ on $\ket{sk}$ to obtain $\ket{\psi_s}$. 
\end{itemize}
Output $(\ket{\psi_s},s)$.

\item \textbf{Multi-Copy Security:} For any polynomial-time quantum adversary $\mathcal{A}$:
\begin{align}
\Pr\left[\begin{array}{c}
(\ket{\psi_s}, s) \leftarrow \mathsf{Mint}(1^{\secparam}) \\
(\sigma_1,\ldots,\sigma_{t+1}) \leftarrow \mathcal{A}(\ket{\psi_s}^{\otimes t}, s) \\
\forall i \in [t+1],\ \mathsf{Ver}(s, \sigma_i) = 1
\end{array}\right] \leq \negl(n)
\end{align}
\end{enumerate}
\end{definition}
\noindent The purity condition in the above definition ensures that for any serial number $s$, the bank can generate multiple copies of $\ket{\psi_s}$. Concretely, the bank can store the secret information $sk$ and to compute $t$ copies of $\ket{\psi_s}$, it can compute $(U\ket{sk})^{\otimes t}$
\par Similar to the mini scheme, we can define a multi-copy security strengthening of~\Cref{def:pkqm} as well. However, since the approach upgrading a mini scheme to a public key quantum money scheme using digital signatures also works in the multi-copy regime, the major challenge is producing a multi-copy secure mini scheme, and so we will omit the definition of a multi-copy secure public key quantum money scheme.

\subsubsection{Copy-Protection}
We recall the definition of quantum copy-protection below. While Aaronson~\cite{Aar09} was the first to define copy-protection, we adopt the subsequent strengthenings of Aaronson's copy-protection definition. 

\begin{definition}[Quantum Copy Protection]
\label{def:qcopy}
A \emph{quantum copy protection scheme} for a family of functions $\mathcal{F}$ consists of two QPT algorithms $(\mathsf{CopyProtect},\allowbreak \mathsf{Eval})$:

\begin{itemize}
    \item $\mathsf{CopyProtect}(1^{\secparam}, f)$: It takes as input a security parameter $\secparam$ and a function $f \in \mathcal{F}$, and outputs a quantum state $\rho_f$ called a \emph{copy-protected program}.
    
    \item $\mathsf{Eval}(\rho_f, x)$: It takes as input a quantum state $\rho_f$ and an input $x$, and outputs $f(x)$.
\end{itemize}
\noindent We require the following properties:\\

\noindent \textbf{Correctness:} For any $f \in \mathcal{F}$ and any input $x$ in the domain of $f$:
\begin{align}
\Pr[\mathsf{Eval}(\mathsf{CopyProtect}(1^{\secparam}, f), x) \rightarrow f(x)] \geq 1 - \mathsf{negl}(\secparam)
\end{align}

\noindent \textbf{Security:} For any polynomial-time quantum adversary $(\mathcal{A}, \mathcal{B}, \mathcal{C})$ and any function $f \in \mathcal{F}$ with input length $n$:
\begin{align}
\Pr\left[\begin{array}{c}
\rho_f \leftarrow \mathsf{CopyProtect}(1^{\secparam}, f) \\
\sigma_{\regB \regC} \leftarrow \mathcal{A}(\rho_f) \\
x_B, x_C \leftarrow \{0,1\}^n \\
y_B \leftarrow \mathcal{B}(\sigma_{\regB}, x_B) \\
y_C \leftarrow \mathcal{C}(\sigma_{\regC}, x_C) \\
y_B = f(x_B) \text{ and } y_C = f(x_C)
\end{array}\right] \leq \mathsf{negl}(\secparam)
\end{align}

where $\sigma_{\regB \regC}$ is a bipartite quantum state shared between $\mathcal{B}$ and $\mathcal{C}$, and the probability is taken over the randomness of $\mathsf{CopyProtect}$, $\mathcal{A}$, and the choice of inputs $x_B, x_C$.
\end{definition}

\paragraph{Multi-copy security.} We strengthen the above definition in two ways. Firstly, we consider the case when the adversary receives many {\em i.i.d} copies of the copy protected state. Next, we consider the case when the adversary receives many {\em identical} copies of the copy protected state. 

\begin{definition}[i.i.d copy security]
A quantum copy protection scheme $(\mathsf{CopyProtect},\allowbreak \mathsf{Eval})$ for a family of functions $\mathcal{F}$ (\Cref{def:qcopy}) is said to satisfy multi-copy security if: 
\begin{itemize}
    \item \textbf{i.i.d copy security:} For any polynomial-time quantum adversary $(\mathcal{A}, \mathcal{B}_1,\ldots,\mathcal{B}_{t+1})$ and any function $f \in \mathcal{F}$ with input length $n$:
\begin{align}
\Pr\left[\begin{array}{c}
\rho_f^{\otimes t} \leftarrow (\mathsf{CopyProtect}(1^{\secparam}, f))^{\otimes t} \\
\sigma_{\regB_1 \cdots \regB_{t+1}} \leftarrow \mathcal{A}(\rho_f^{\otimes t}) \\
x_{B_1},\ldots,x_{B_{t+1}} \leftarrow \{0,1\}^n \\
\forall i \in [t+1], y_{B_i} \leftarrow \mathcal{B}_i(\sigma_{\regB_i}, x_{B_i}) \\
\forall i \in [t+1], y_{B_i} = f(x_{B_i})
\end{array}\right] \leq \mathsf{negl}(\secparam)
\end{align}

where $\sigma_{\regB_1 \cdots \regB_{t+1}}$ is a $t$-partite quantum state shared amongst $\regB_1,\ldots,\regB_{t+1}$, and the probability is taken over the randomness of $\mathsf{CopyProtect}$, $\mathcal{A}$, and the choice of inputs $x_{B_1},\ldots,x_{B_{t+1}}$.
\end{itemize}
\end{definition}

\noindent Prior works~\cite{LLLQ22,CG24b} showed the existence of copy-protection for some cryptographic functionalities satisfying i.i.d copy security.  
\par Similarly, we can define identical copy security as follows: 

\begin{definition}[Identical copy security]
A quantum copy protection scheme $(\mathsf{CopyProtect},\allowbreak \mathsf{Eval})$ for a family of functions $\mathcal{F}$ (\Cref{def:qcopy}) is said to satisfy multi-copy security if: 
\begin{itemize}
    \item \textbf{Purity}: the output of $\mathsf{CopyProtect}$ is a pure state. That is, $\mathsf{Copyprotect}(1^{\secparam},f)$ outputs $\ket{\psi_f}$. 
    \item \textbf{Identical copy security:} For any polynomial-time quantum adversary $(\mathcal{A}, \mathcal{B}_1,\ldots,\mathcal{B}_{t+1})$ and any function $f \in \mathcal{F}$ with input length $n$:
\begin{align}
\Pr\left[\begin{array}{c}
\ket{\psi_f} \leftarrow (\mathsf{CopyProtect}(1^{\secparam}, f)) \\
\sigma_{\regB_1 \cdots \regB_{t+1}} \leftarrow \mathcal{A}(\rho_f^{\otimes t}) \\
x_{B_1},\ldots,x_{B_{t+1}} \leftarrow \{0,1\}^n \\
\forall i \in [t+1], y_{B_i} \leftarrow \mathcal{B}_i(\sigma_{\regB_i}, x_{B_i}) \\
\forall i \in [t+1], y_{B_i} = f(x_{B_i})
\end{array}\right] \leq \mathsf{negl}(\secparam)
\end{align}

where $\sigma_{\regB_1 \cdots \regB_{t+1}}$ is a $t$-partite quantum state shared amongst $\regB_1,\ldots,\regB_{t+1}$, and the probability is taken over the randomness of $\mathsf{CopyProtect}$, $\mathcal{A}$, and the choice of inputs $x_{B_1},\ldots,x_{B_{t+1}}$.
\end{itemize}
\end{definition}
\fi
\section{Main Theorem - Simulating families of mixed states}

We state our main theorem. \pnote{need to explicitly state the value of $t$ below.} \pnote{discuss the main theorem and give some context for it.}

\begin{theorem}[Main Theorem]
    \label{thm:mainthm}
    Let $n_A,n_B,n_C,t\in \N$. Then there is a CPTP map $\Sim$ sending $t\cdot n_A$ to $t(n_A+n_B+n_C)$ qubits running in time $\poly(n_A,n_B,n_C)$ such that the following holds. 
    
    Let $N\in \N$ and suppose $\{\sigma_i\}_{i\in [N]}$ is a family of mixed with corresponding purifications $\{\ket{\phi_i}_{\regA,\regB}\}_{i\in [N]}$ on registers $\regA,\regB$ of dimensions $2^{n_A},2^{n_B}$ respectively. Let $q = 2t$. There exists a family of states $\{\ket{\psi_{f_1,f_2,f_3,f_4}}_{\regA,\regB}\}$ for $f_1:\{0,1\}^{n_C}\to [q],f_2:\{0,1\}^{n_C}\to \{0,1\}^{n_B},f_3:\{0,1\}^{n_C}\to \{0,1\}^{n_B}$ and $f_4:\{0,1\}^{n_C} \rightarrow [N]$ such that for all $t \leq \frac{q}{2}$, the following hold
    \begin{enumerate}
        \item $t$ copies of a random sample $\ket{\psi_k}^{\otimes t}$ can be efficiently approximated using $t$ i.i.d. copies of $\sigma_i$. That is, on input
        $$\E_{i_1,\dots,i_t}[\sigma_{i_1}\otimes \dots \otimes \sigma_{i_t}]$$
        $\Sim$ outputs a state $\rho$ such that
        $$\TD\left(\rho,\E_{f_1,f_2,f_3,f_4}[\ketbra{\psi_{f_1,f_2,f_3,f_4}}^{\otimes t}]\right) \leq \frac{t^2}{2^{n_C}}$$
        \item $\{\ket{\psi_k}\}_{k\in \mathcal{K}}$ is a purification of $\{\sigma_i\}_{i\in \{0,1\}^n}$. In particular, 
        $$\E_k[\Tr_{\regB}(\ketbra{\phi_k})] = \E_i[\sigma_i]$$
        For the case when $N=1$ and $\{\sigma_i\}$ consists of a single mixed state $\sigma$, we further have that for all $k\in \mathcal{K}$,
        $$\Tr_{\regB}[\ketbra{\phi_k}] = \sigma$$
        \item Given query access to $f_1,f_2,f_3,f_4$ as well as a circuit preparing $\{\ket{\phi_i}\}_{i\in [N]}$, $\ket{\psi_{f_1,f_2,f_3,f_4}}$ is preparable in time $\poly(\log_2 N,n_A,n_B,n_C)$.
    \end{enumerate}
\end{theorem}

We will construct the simulator $\Sim$ and the family $\{\ket{\psi_{f_1,f_2,f_3,f_4}}\}$ as follows:

\begin{definition}[Purification ensemble construction]
\label{cons:mainthm}
    Let $\{\sigma_i\}_{i\in [N]},\{\ket{\phi_i}_{\regA,\regB}\}_{i\in [N]},n_A,n_B,n_C$ and the support of $f_1,f_2,f_3,f_4$ be as in~\Cref{thm:mainthm}. Then we define
    $$\ket{\psi_{f_1,f_2,f_3,f_4}}_{\regA \regB \regC} =  \sum_{i\in \{0,1\}^{n_C}} \left(I_{\regA} \otimes X^{f_2(i)}_{\regB} Z^{f_3(i)}_{\regB} \otimes I_{\regC} \right) \frac{\omega_q^{f_1(i)}}{\sqrt{2^{n_C}}} \ket{\phi_{f_4(i)}}_{\regA \regB}\otimes\ket{i}_{\regC}$$
\end{definition}

\begin{definition}[Simulator construction]
    \label{cons:sim}
    Let $n_A,n_B,n_C\in \N$. For $i\in [t]$, let $\ket{\chi_i}$ be any state on $n_A+n_B+n_C$ qubits. On input $\ket{\chi_1},\dots,\ket{\chi_t}$, $\Sim$ will behave as follows
    \begin{itemize}
        \item First, sample distinct $r_1,\dots,r_t$ uniformly at random from $\{0,1\}^{n_B+n_C}$.
        \item Output
        $$\sum_{\pi\in Sym([t])} \bigotimes_{j=1}^t \ket{\chi_{\pi(i)}}_{\regA_j}\ket{r_{\pi(i)}}_{\regB_j \regC_j}$$
    \end{itemize}
\end{definition}

Note that properties $2$ and $3$ of $\ket{\psi_{f_1,f_2,f_3,f_4}}$ are self-evident. Thus, the proof of this theorem boils down to the following key lemma as well as the fact that $\sims$ can be implemented efficiently.

\begin{theorem}
\label{lem:main}
Let $N,n_A,n_B,n_C,t \in \N$ and set $q=2t$. Suppose $\{\ket{\phi_{j}}\}_{j \in [N]}$ is a family of quantum states on registers $\regA,\regB$, of dimension $2^{n_A}$ and $2^{n_B}$ respectively. Let $f_1:\{0,1\}^{n_C}\to [q],f_2:\{0,1\}^{n_C}\to \{0,1\}^{n_B},f_3:\{0,1\}^{n_C}\to \{0,1\}^{n_B}$ and $f_4:\{0,1\}^{n_C} \rightarrow [N]$. Let $\{\ket{\psi_{f_1,f_2,f_3,f_4}}\}$ be as in~\Cref{cons:mainthm}.

Let $\sims$ be as in~\Cref{cons:sim}.

\noindent Define $\rho_t$ as follows:
    $$\rho_t = \E_{f_1,f_2,f_3,f_4}\left[\ketbra{\psi_{f_1,f_2,f_3,f_4}}^{\otimes t}\right]$$

    \noindent Let 
    $$\rho' = \sims\left( \left(\E_j\left[\Tr_{\regB}(\ketbra{\phi_j}_{\regA \regB})\right]\right)^{\otimes t}\right)$$
    Then $\TD(\rho',\rho_t) \leq \frac{t^2}{2^{n_C}}$. 
\end{theorem}

\begin{proof}
\noindent Similar to the proof structure presented in the technical overview~(\Cref{sec:techoverview}), we divide the proof into two parts. In the first part, we perfom the analysis for the case when the controlled one-time pad is not applied. While in~\Cref{sec:techoverview} an intuitive proof via the compressed oracle method was presented, we present a direct proof below. In the second part, we consider the action of the controlled quantum one-time pad. 

\paragraph{Part I: Ignoring the controlled one-time pad.} Define the following state: 
$$\ket{\psi_{f_1,f_4}}_{\regA \regB \regC} = \frac{1}{\sqrt{2^{n_C}}} \sum_{i\in \{0,1\}^{n_C}}\omega_q^{f_1(i)} \ket{\phi_{f_4(k)}}_{\regA \regB}\otimes\ket{i}_{\regC}$$

\noindent Fix $\bm{i}=(i_1,\ldots,i_t) \in \{0,1\}^{n_C \cdot t}$ and $\bm{i'}=(i'_1,\ldots,i'_t) \in \{0,1\}^{n_C \cdot t}$.  

Define $\rho_t^{(f_1,f_4)}[\bm{i},\bm{i'}]$ as follows:
$$\rho_t^{(f_1,f_4)}[\bm{i},\bm{i'}] = \E_{f_1,f_4} \left[ \omega_q^{\sum_{j \in [t]}f_1(i_j)-f_1(i'_j)} \bigotimes_{j \in [t]} \ketbra{\phi_{f_4(i_j)}}{\phi_{f_4(i'_j)}}_{\regA_j \regB_j}\otimes\ketbra{i_j}{i'_j}_{\regC_j} \right]$$

\noindent Note that for any fixed $i_1,\ldots,i_t,i'_1,\ldots,i'_t$, we have that $\E_{f_1}\left[\omega_q^{\sum_j \left(f_1(i_j) - f_1(i'_j)\right)}\right]=1$ if $\type((i_1,\ldots,i_t))=\type((i'_1,\ldots,i'_t))$ and $\E_{f_1}\left[\omega_q^{\sum_j \left(f_1(i_j) - f_1(i'_j)\right)}\right]=0$, otherwise.

\noindent Thus, we have the following:
\begin{enumerate}
    \item For $\bm{i},\bm{i'}$ such that $\type((i_1,\dots,i_t))=\type((i_1',\dots,i_t'))$
    $$\rho_t^{(f_1,f_4)}[\bm{i},\bm{i'}]=\E_{f_4} \left[  \bigotimes_{j \in [t]} \ketbra{\phi_{f_4(i_j)}}{\phi_{f_4(i'_j)}}_{\regA_j \regB_j}\otimes\ketbra{i_j}{i'_j}_{\regC_j} \right]$$
    \item For $\bm{i},\bm{i'}$ such that $\type((i_1,\dots,i_t))\neq\type((i_1',\dots,i_t'))$
    $$\rho_t^{(f_1,f_4)}[\bm{i},\bm{i'}]=0$$
\end{enumerate}

\noindent We define $\distinct(n,t)$ to be the following set: $\{(i_1,\ldots,i_t)\ :\ \forall j \neq j',\ i_j \neq i_{j'}\}$. \\

\paragraph{Part II: Action of controlled quantum one-time pad.} Define $$U_{f_2,f_3,\bm{i}}=\left( \bigotimes_{j \in [t]} I_{\regA_j} \otimes X_{\regB_j}^{f_2(i_j)} Z_{\regB_j}^{f_3(i_j)} \otimes I_{\regC_j} \right).$$ Define $\{\ket{\chi_{h}^{\ell}}\}_{\ell}$ to be the eigenbasis in the spectral decomposition of $\Tr_{\regB}(\ketbra{\phi_h}_{\regA \regB})$. In more detail, define $\Tr_{\regB}(\ketbra{\phi_h}_{\regA \regB}) = \E_{\ell \leftarrow {\cal D}_h}\left[ \ketbra{\chi_{h}^{\ell}} \right]$ for some distribution ${\cal D}_h$. \par We define $P_{\pi}$ to be a permutation operator that permutes the blocks of qubits. That is, $P_{\pi}$ acts on all the registers $(\regA_1,\regB_1,\regC_1,\ldots,\regA_t,\regB_t,\regC_t)$ and permutes the contents of all the blocks (the $j^{th}$ block is comprised of $(\regA_j,\regB_j,\regC_j)$) according to the permutation $\pi$. Consider the following: 

\allowdisplaybreaks
\begin{eqnarray*}
 & \rho_t & \\
& = & \E_{\substack{f_2,f_3\\ \bm{i}=(i_1,\ldots,i_t)\\ \bm{i'}=(i'_1,\ldots,i'_{t})}}\left[ U_{f_2,f_3,\bm{i}}\ \rho_t^{(f_1,f_4)}[\bm{i},\bm{i'}]\  U_{f_2,f_3,\bm{i'}}^{\dagger} \right]  \\
    & = & \E_{\substack{f_2,f_3\\ \bm{i}=(i_1,\ldots,i_t)\\ \bm{i'}=(i'_1,\ldots,i'_{t})\\ \type(\bm{i})=\type(\bm{i'})}} \left[  U_{f_2,f_3,\bm{i}} \E_{f_4} \left[  \bigotimes_{j \in [t]} \ketbra{\phi_{f_4(i_j)}}{\phi_{f_4(i'_j)}}_{\regA_j \regB_j}\otimes\ketbra{i_j}{i'_j}_{\regC_j} \right]  U_{f_2,f_3,\bm{i'}}^{\dagger} \right] \\
    & \approx_{\frac{t^2}{2^{n_C}}} & \E_{\substack{f_2,f_3\\ \bm{i}=(i_1,\ldots,i_t) \in \distinct(n,t)\\ \bm{i'}=(i'_1,\ldots,i'_{t}) \in \distinct(n,t)\\ \type(\bm{i})=\type(\bm{i'})}} \left[  U_{f_2,f_3,\bm{i}} \E_{f_4} \left[  \bigotimes_{j \in [t]} \ketbra{\phi_{f_4(i_j)}}{\phi_{f_4(i'_j)}}_{\regA_j \regB_j}\otimes\ketbra{i_j}{i'_j}_{\regC_j} \right]  U_{f_2,f_3,\bm{i'}}^{\dagger} \right] \\
    & = & \E_{\substack{f_2,f_3\\ \bm{i}=(i_1,\ldots,i_t) \in \distinct(n,t)\\ \pi \in S_t}} \left[  U_{f_2,f_3,\bm{i}} \E_{f_4} \left[  \bigotimes_{j \in [t]} \ketbra{\phi_{f_4(i_j)}}{\phi_{f_4(i_j)}}_{\regA_j \regB_j}\otimes\ketbra{i_j}{i_j}_{\regC_j} \right]  U_{f_2,f_3,\bm{i}}^{\dagger} P_{\pi} \right] \\
    & = & \E_{\substack{f_4\\ \bm{i}=(i_1,\ldots,i_t) \in \distinct(n,t)\\ \pi \in S_t}} \left[  \left( \bigotimes_{j \in [t]} \Tr_{\regB_j}\left( \ketbra{\phi_{f_4(i_j)}}{\phi_{f_4(i_j)}}_{\regA_j \regB_j} \right)\otimes  \frac{I}{2^{|\regB_j|}} \otimes \ketbra{i_j}{i_j}_{\regC_j} \right) P_{\pi} \right] \\
    & = & \E_{\substack{f_4\\ \bm{i}=(i_1,\ldots,i_t) \in \distinct(n,t)\\ \bm{k}=(k_1,\ldots,k_t) \in \distinct(n,t) \\ \pi \in S_t}} \left[  \left( \bigotimes_{j \in [t]} \Tr_{\regB_j}\left( \ketbra{\phi_{f_4(i_j)}}{\phi_{f_4(i_j)}}_{\regA_j \regB_j} \right)\otimes  \ketbra{k_j}_{\regB_j} \otimes \ketbra{i_j}{i_j}_{\regC_j} \right) P_{\pi} \right] \\ 
     & = & \E_{\substack{f_4\\ \bm{i}=(i_1,\ldots,i_t) \in \distinct(n,t)\\ \bm{\ell}=(\ell_1,\ldots,\ell_t) \gets (\mathcal{D}_{h_1},\dots,\mathcal{D}_{h_t})\\ \bm{k}=(k_1,\ldots,k_t) \in \distinct(n,t) \\ \pi \in S_t}} \left[  \left( \bigotimes_{j \in [t]} 
     \ketbra{\chi_{f_4(i_j)}^{\ell_j}}_{\regA_j} \otimes  \ketbra{k_j}_{\regB_j} \otimes \ketbra{i_j}{i_j}_{\regC_j} \right) P_{\pi} \right] \\
     & = & \E_{\substack{f_4\\ \bm{i}=(i_1,\ldots,i_t) \in \distinct(n,t)\\ \bm{i'}=(i'_1,\ldots,i'_t) \in \distinct(n,t)\\ \bm{\ell}=(\ell_1,\ldots,\ell_t) \gets (\mathcal{D}_{h_1},\dots,\mathcal{D}_{h_t})\\ \bm{\ell'}=(\ell'_1,\ldots,\ell'_t) \gets (\mathcal{D}_{h_1},\dots,\mathcal{D}_{h_t})\\ \bm{k}=(k_1,\ldots,k_t) \in \distinct(n,t) \\ \bm{k'}=(k'_1,\ldots,k'_t) \in \distinct(n,t)}} \left[  \left( \bigotimes_{j \in [t]} 
     \ketbra{\chi_{f_4(i_j)}^{\ell_{j}}}{\chi_{f_4(i_j)}^{\ell'_{j}}}_{\regA_j} \otimes  \ketbra{k_j}{k'_j}_{\regB_j} \otimes \ketbra{i_j}{i'_j}_{\regC_j} \right) \right] \\
     & = & \E_{\substack{(h_1,\ldots,h_t) \xleftarrow{\$} [N] \\ \bm{\chi} \leftarrow \bigtimes_{j} \{\chi_{h_j}^{\ell}\}_{\ell}\\ u \xleftarrow{\$} \{0,1\}^{2^{n_B+n_C}}:\\ {\sf wt}(u)=t }} \left[  \ketbra{\setst_{\bm{\chi},t,u}} \right] = \rho' 
\end{eqnarray*}
\noindent We have shown so far that $\rho_t \approx_{\frac{t^2}{2^{n_C}}} \rho'$. \\  

\noindent \pnote{add an informal description of the simulator}

\paragraph{Efficient implementation of ${\sf Sim}$.} We will now show that there is an efficient algorithm $\sims$ that takes as input $\left( \ket{\chi_{h_1}^{\ell_1}},\ldots,\ket{\chi_{h_t}^{\ell_t}} \right)$, where $h_i$ are picked uniformly at random and $\ell_i$ are sampled from $\mathcal{D}_{h_i}$, and outputs the state $\rho'$. \\

\noindent $\sims$ does the following:
\begin{enumerate}
    \item The input state $\ket{\chi_{h_j}^{(\ell_j)}}$ is initialized in the register $\regD_j$. 
    \item It samples $i_1,\dots,i_t$ uniformly at random from $\{0,1\}^{n_B+n_C}$ subject to the condition that they are all distinct.
    \item It efficiently generates the state $\frac{1}{\sqrt{t!}} \sum_{\pi \in S_t} \ket{\pi} \ket{i_{\pi(1)}}_{\regB_1 \regC_1}\cdots \ket{i_{\pi(t)}}_{\regB_t\regC_t}$. Controlled on the first register containing $\pi$, it then prepares the following state:
    $$\frac{1}{\sqrt{t!}} \sum_{\pi \in S_t} \ket{\pi} \ket{\chi_{h_{\pi(1)}}^{(\ell_{\pi(1)})}}_{\regA_1} \ket{i_{\pi(1)}}_{\regB_1 \regC_1}\cdots \ket{\chi_{h_{\pi(t)}}^{(\ell_{\pi(t)})}}_{\regA_t}\ket{i_{\pi(t)}}_{\regB_t\regC_t}$$
    
    \item Finally it then uncomputes the first register using $(i_1,\ldots,i_t)$ to get the following state which is output by the algorithm: 
    $$\ket{\setst_{\bm{\chi},t,u}} = \frac{1}{\sqrt{t!}} \sum_{\pi \in S_t} \ket{\chi_{h_{\pi(1)}}^{(\ell_{\pi(1)})}}_{\regA_1} \ket{i_{\pi(1)}}_{\regB_1 \regC_1}\cdots \ket{\chi_{h_{\pi(t)}}^{(\ell_{\pi(t)})}}_{\regA_t}\ket{i_{\pi(t)}}_{\regB_t\regC_t},$$
    where $\bm{\chi}=\left( \ket{\chi_{h_1}^{(\ell_1)}},\ldots,\ket{\chi_{h_t}^{(\ell_t)}} \right)$ and $u \in \{0,1\}^{\cdot 2^{n_B+n_C}}$ such that $u_{\ell}=1$ if and only if $\ell \in \{i_1,\ldots,i_t\}$.   
\end{enumerate}

\end{proof}

\ifnum\llncs=0
\section{Multi-Copy Secure Unclonable Cryptography}

\newcommand{\design}{{\sf design}}
\newcommand{\pha}{{\sf phase}}
\newcommand{\sig}{{\sf sign}}
\newcommand{\fclass}{{\cal F}}
\newcommand{\copyprotect}{{\sf CopyProtect}}
\newcommand{\Eval}{{\sf Eval}}
\newcommand{\smcp}{{\sf cp}}
\newcommand{\alice}{{\cal A}}
\newcommand{\bob}{{\cal B}}
\newcommand{\charlie}{{\cal C}}
\newcommand{\bfadv}{{\bf \adversary}}

\subsection{Public-Key Quantum Money}

We will show
\begin{theorem}\label{thm:ismini}
    If quantum secure one-way functions exist and there exists a one-copy secure mini scheme, then there exists a multi-copy secure mini scheme.
\end{theorem}

We first construct a multi-copy secure mini scheme (\Cref{def:multicopy:mini-scheme}). The transformation from mini scheme (\Cref{def:mini-scheme}) to full fledged public-key money scheme (\Cref{def:pkqm}) using digital signatures preserves the multi-copy security. That is, assuming post-quantum secure one-way functions, there exists a {\em multi-copy} secure public-key money scheme assuming {\em multi-copy} secure mini schemes. 

\paragraph{Starting point: public key quantum money} We will begin with a public key quantum money scheme $\gen,\mint,\verify$. Note that we can construct a public-key quantum money scheme from a one-copy secure mini scheme as well a one-way function, so assuming one-way functions building this is free.

\paragraph{Multi-copy secure mini scheme} We will construct another mini scheme $(\mint',\allowbreak \verify')$ such that even given $t$ copies of the money state produced by $\mint'$, where $t$ is an arbitrary polynomial, any computationally bounded adversary cannot produce $t+1$ copies of the state that passes $\verify'$. To design this new mini scheme, we will use a post-quantum secure pseudorandom function $f:\{0,1\}^{\secparam} \times \{0,1\}^{n+1} \rightarrow \{0,1\}^m$.

\begin{itemize}
    \item $\mint'(1^{\secparam})$: it takes as input a security parameter $\secparam$,   
    \begin{itemize}
        \item It first executes $\gen(1^{\secparam})$ to obtain $(\pk,\sk)$.
        \item Let $\ket{\phi^{\sk}}_{\regM,\regS,\regB}$ be the purification of the state produced by $\mint(\sk)$ with money state in register $\regM$, serial number in register $\regS$ of dimension $2^{n_S}$, and purification register $\regB$ of dimension $2^{n_B}$. In particular, $\Tr_{\regB}(\ketbra{\phi^{\sk}})$ will be the same density matrix as $\mint(\sk)$, including both the money state and serial number.
        \item It samples three PRF keys $k_1,k_2,k_3 \xleftarrow{\$} \{0,1\}^{\secparam}$ for three pseudorandom functions $f^1:\{0,1\}^\secparam\to [2^{\secparam}],f^2:\{0,1\}^\secparam\to \{0,1\}^{n_B},f^3:\{0,1\}^\secparam \to \{0,1\}^{n_B}$ on appropriate domains.
        \item Let $\ket{\psi_{f_1,f_2,f_3}^{\sk}}$ be the construction~\Cref{cons:mainthm} applied to the purification $\{\ketbra{\phi^{\sk}}_{\regM,\regS,\regB}\}$ (we omit $f_4$ since it is required to be the constant function) with $q = 2^\secparam$.
        \item Finally, it outputs a pure money state $\ket{\$}=\ket{\psi_{f^1_{k_1},f^2_{k_2},f^3_{k_3}}^{\sk}}$ and a serial number $s = \pk$. This is efficiently generatable by~\Cref{thm:mainthm}, as $\mint'$ can efficiently compute oracle queries to $f^1_{k_1},f^2_{k_2},f^3_{k_3}$.
    \end{itemize}

    \item $\verify'$: it takes as input $(\pk,\ket{\$}_{\regM \regS \regB \regC})$ and does the following:
    \begin{itemize}
        \item Set $\rho_{\regM,\regS} = \Tr_{\regB,\regC}(\ketbra{\$})$.
        \item Coherently run $\verify(\ketbra{\pk} \otimes \rho_{\regM,\regS})$ and measure the result.
        \item Output $1$ if and only if $\verify$ accepts.
    \end{itemize}
\end{itemize}

We first remark that this is indeed a pure scheme. $\mint$ first samples $s=\pk$ along with secret randomness $\sk'=(\sk,k_1,k_2,k_3)$, and then outputs the state $\ket{\psi_{f_1,f_2,f_3}^{\sk}}$, which is the result of some isometry applied to $\sk'$. We proceed to prove~\Cref{thm:ismini} by showing correctness and multi-copy security hold.

\begin{theorem}
    Assuming $(\gen,\mint,\verify)$ satisfies correctness as a public key quantum money protocol, the above mini scheme $(\mint',\verify')$ satisfies correctness.
\end{theorem}

\begin{proof}
    Note that by~\Cref{thm:mainthm}, for all $k_1,k_2,k_3$, $\Tr_{\regB,\regC}(\ketbra{\psi^{\sk}_{k_1,k_2,k_3}}) = \Tr_{\regB}(\ketbra{\phi^{\sk}}) = \mint(\sk)$. And so correctness follows from correctness of public key quantum money.
\end{proof}

\begin{theorem}
    Assuming $(\gen,\mint,\verify)$ is a secure public key quantum money protocol, the above mini scheme $(\mint',\verify')$ is multi-copy secure.
\end{theorem}

\begin{proof}
    Suppose there exists a QPT adversary $\A$ that violates the multi-copy security of $(\mint',\verify')$. That is, for random $k_1,k_2,k_3$, $\gen\to (\sk,\pk)$, 
    $$\BigPr{(1,\dots,1)\gets (\verify')^{\otimes (t+1)}(\pk,\rho)}{k_1,k_2,k_3\gets \{0,1\}^\secparam\\\mint(1^\secparam)\to (\sk,\pk)\\
    \sigma \gets\ketbra{\psi_{f^1_{k_1},f^2_{k_2},f^3_{k_3}}^{\sk}}^{\otimes t} \\\A\left(\pk,\sigma\right) \to \rho}\geq \frac{1}{\secparam^c}$$
    for some constant $c$. We will show how to convert $\A$ into a QPT adversary $\reduction$ that can violate the security of $(\gen,\mint,\verify)$. In particular, let $\Sim$ be the simulator from~\Cref{cons:sim} for $n_A = n_M+n_S,n_B=n_B,n_C=\secparam$. Then, $\reduction(\pk,\mint(\sk)^{\otimes t})$ will do the following
    \begin{enumerate}
        \item Run $\A(\Sim(\mint(\sk)^{\otimes t}),\pk) \to \rho_{\regM_1\regS_1\regB_1\regC_1\dots\regM_{t+1}\regS_{t+1},\regB_{t+1}\regC_{t+1}}$
        \item Measure registers $\regS_1,\dots,\regS_{t+1}$ in the standard basis
        \item Output $\Tr_{\regB_1 \regC_1\dots \regB_{t+1}\regC_{t+1}}(\rho)$
    \end{enumerate}.

    By pseudorandomness of $f^1,f^2,f^3$, we have
    $$\BigPr{(1,\dots,1)\gets (\verify')^{\otimes (t+1)}(\pk,\rho)}{f_1,f_2,f_3\xleftarrow{\$}\\\gen(1^\secparam)\to (\sk,\pk)\\
    \sigma \gets\ketbra{\psi_{f_1,f_2,f_3}^{\sk}}^{\otimes t} \\\A\left(\pk,\sigma\right) \to \rho}\geq \frac{1}{\secparam^c} - \negl(\secparam)$$

    By~\Cref{thm:mainthm}, we then have
    $$\BigPr{(1,\dots,1)\gets (\verify')^{\otimes (t+1)}(\pk,\rho)}{\gen(1^\secparam)\to (\sk,\pk)\\
    \sigma \gets\Sim(\mint(\sk)^{\otimes t}) \\\A\left(\pk,\sigma\right) \to \rho}\geq \frac{1}{\secparam^c} - \negl(\secparam)$$

    But this is exactly $\reduction$'s advantage in breaking $(\gen,\mint,\verify)$, and so we are done.
\end{proof}

\noindent Instantiating $(\mint,\verify)$ with Zhandry's quantum mini scheme~\cite{Zhandry19}, we have the following consequence: 

\begin{corollary}
Assuming post-quantum secure indistinguishability obfuscation and post-quantum secure injective one-way functions, there exist a multi-copy secure public-key quantum money mini scheme (\Cref{def:multicopy:mini-scheme}). 
\end{corollary}

\subsection{Copy-Protection}

\paragraph{Starting point: i.i.d copy secure copy-protection scheme.} Suppose there exists a i.i.d copy secure copy-protection scheme $(\copyprotect,\Eval)$ for $\fclass$ with the following structure: the copy-protection algorithm $\copyprotect$, on input $(1^{\secparam},f)$, first samples $r \xleftarrow{\$} \{0,1\}^{\ell(\secparam)}$ uniformly at random, applies $U_{\smcp}^{(\secparam,f)}$ on $\ket{r}_{\regX}\ket{0}_{\regY}$ to obtain $\ket{\psi_{r}^{(\secparam,f)}}_{\regA \regB}$. It traces out the register $\regB$ and outputs the register $\regA$ as the copy-protected state. We will assume that without loss of generality, $\Eval$ first applies a unitary $U_{\Eval}$ followed by measuring the first $m$ qubits, where $m$ is the output length of $f$. 
\par Most of the copy-protection schemes proposed in the literature~\cite{CLLZ22,LLLQ22,AB24,CG24,ABZ25,KY25} proceed by first sampling classical randomness and then deterministically computing the copy-protected state. In particular, the i.i.d secure copy-protection schemes proposed in the literature~\cite{LLLQ22,CG24} present an instantiation of the above template.  

\paragraph{Construction.} In addition to $(\copyprotect,\Eval)$ will also use a post-quantum secure pseudorandom function $f:\{0,1\}^{\secparam} \times \{0,1\}^{n+1} \rightarrow \{0,1\}^m$.
\par We show that there exists a multi-copy secure copy-protection scheme for $\fclass$.
\begin{itemize}
    \item $\copyprotect'\left(1^{\secparam},F\right)$: it does the following: 
    \begin{itemize}
        \item It prepares a uniform superposition over $\ell$-bit strings: $\sum_{r \in \{0,1\}^{\ell}} \frac{1}{\sqrt{2^{\ell}}} \ket{r}_{\regR}$, 
        \item It samples two PRF keys $k_1,k_2 \xleftarrow{\$} \{0,1\}^{\secparam}$. It applies the unitary $U_{\pha} = \sum_{r} \omega_q^{f(k_1,r||0)} \ketbra{r}_{\regR}$, where $q = 2^{\omega(\log(\secparam))}$. The resulting state is \\$\sum_{r \in \{0,1\}^{\ell}} \frac{\omega_q^{f(k_1,r)}}{\sqrt{2^{\ell}}} \ket{r}_{\regR}$
        \item It copies $\regR$ onto a different register $\regX$. It initalizes $\regY$ with $\ket{0}$. The resulting state is $\sum_{r \in \{0,1\}^{\ell}} \frac{\omega_q^{f(k_1,r)}}{\sqrt{2^{\ell}}} \ket{r}_{\regR} \ket{r}_{\regX} \ket{0}_{\regY}$. 
        \item It applies $I_{\regR} \otimes U_{\regX \regY}^{(\secparam,F)}$ to obtain the state $\sum_{r \in \{0,1\}^{\ell}} \frac{\omega_q^{f(k_1,r)}}{\sqrt{2^{\ell}}} \ket{r}_{\regR} \ket{\psi_r^{(\secparam,F)}}_{\regA \regB}$.
        \item  It applies the unitary $U_{\design} = \sum_{r} \ketbra{r}_{\regR} \otimes I_{\regA} \otimes X^{f(k_2,r||0)}_{\regB} Z^{f(k_2,r||1)}_{\regB}$ to obtain the following state:
        $$\ket{\Psi_F} = \sum_{r \in \{0,1\}^{\ell}} \frac{\omega_q^{f(k_1,r)}}{\sqrt{2^{\ell}}} \left( I_{\regR} \otimes I_{\regA} \otimes X^{f(k_2,r||0)}_{\regB} Z^{f(k_2,r||1)}_{\regB} \right) \ket{r}_{\regR} \ket{\psi_r^{(\secparam,F)}}_{\regA \regB}$$
        \item Output $\ket{\Psi_F}_{\regR \regA \regB}$. 
    \end{itemize}
    \item $\Eval'(\ket{\Psi_F}_{\regR \regA \regB},x)$: 
    \begin{itemize}
        \item It applies $I_{\regR} \otimes (U_{\Eval})_{\regA} \otimes I_{\regB}$,
        \item It measures the first $m$ qubits of $\regA$ to obtain the outcome $y$,
        \item Output $y$.
    \end{itemize}
\end{itemize}

\noindent Assuming the post-quantum security of the pseudorandom function $f$, the correctness of the above construction follows from the correctness of the copy-protection scheme. 

\begin{theorem}
Assuming $(\copyprotect,\Eval)$ satisfies iid muli-copy copy-protection security, $f$ is a post-quantum secure pseudorandom function, $(\copyprotect',\Eval')$ satisfies multi-copy security. 
\end{theorem}
\begin{proof}
Suppose there exists a QPT adversary $\bm{\alice'}=(\alice',\bob'_1,\ldots,\bob'_{t+1})$ that violates the multi-copy security of $(\copyprotect',\Eval')$. That is, consider the following security experiment: 
\begin{itemize}
    \item $\alice'$ gets $t$ copies of $\ket{\Psi_f}_{\regR \regA \regB}$ and outputs a $(t+1)$-partite state on registers $\regZ_{\bob'_1},\ldots,\regZ_{\bob'_{t+1}}$. It sends the register $\regZ_{\bob'_i}$ to $\bob'_i$. 
    \item $\bob'_i$ then gets as input $x_i \xleftarrow{\$} \{0,1\}^n$, where $n$ is the input length of $f$. It then outputs $(y_1,\ldots,y_{t+1})$. 
\end{itemize}
\noindent We denote the probability that $(y_1,\ldots,y_{t+1})=(f(x_1),\ldots,f(x_{t+1}))$ to be $p$, where $p$ is non-negligible. Using this, we design a QPT adversary ${\bm \alice}=(\alice,\bob_1,\ldots,\allowbreak \bob_{t+1})$ that violates the iid multi-copy security of the copy-protection scheme $(\copyprotect,\allowbreak \Eval)$.  
\par We prove this by a hybrid argument. \\

\noindent $\hybrid_1$: This corresponds to the real experiment. That is, $\bm{\alice}$ receives as input $t$ copies of $\ket{\psi'}_{\regS \regA \regB \regC}$ and outputs $(y_1,\ldots,y_{t+1})$. We refer to the success probability of $\adversary$ as the probability that $(y_1,\ldots,y_{t+1})=(f(x_1),\ldots,f(x_{t+1}))$, which is $p$. \\

\noindent $\hybrid_2$:  Modify the generation of $\ket{\Psi_f}_{\regR \regA \regB}$ as follows: instead of applying the $U_{\design} = \sum_{r} \ketbra{r}_{\regR} \allowbreak \otimes I_{\regA} \otimes X^{f(k_2,r||0)}_{\regB} Z^{f(k_1,r||1)}_{\regB}$, instead apply the unitary $\sum_{r} \ketbra{r}_{\regR} \otimes I_{\regA} \otimes X^{f(r||0)}_{\regB} Z^{f(r||1)}_{\regB}$, where $f$ is a random function. 
\par From the post-quantum security of pseudorandom function $f$, the hybrids $\hybrid_1$ and $\hybrid_2$ are computationally indistinguishable. The success probability of $\adversary$ in this hybrid is negligibly close to $p$.  
\\

\noindent $\hybrid_3$: Modify the generation of $\ket{\Psi_f}_{\regR \regA \regB}$ as follows: instead of applying the $U_{\pha} = \sum_{s} \omega_q^{f(k_1,s)} \ketbra{s}$, instead apply the unitary $U_{\pha} = \sum_{s} \omega_q^{f(r)} \ketbra{s}$, where $f$ is a random function. 
\par From the post-quantum security of pseudorandom function $f$, the hybrids $\hybrid_2$ and $\hybrid_3$ are computationally indistinguishable. The success probability of $\adversary$ in this hybrid is negligibly close to $p$.\\

\noindent $\hybrid_4$: Let $\sims$ be the efficient algorithm from~\Cref{cons:sim}. Execute\\ $\sims\left( \Tr_{\regB}\left( \ketbra{\Psi_f}_{\regR \regA \regB} \right)^{\otimes t} \right)$ to obtain $\sigma$. Execute $\bm{\alice}(\sigma)$ to obtain $\rho$. 
\par The success probability of $\adversary$ in this hybrid is negligibly close to $p$. This follows from the fact that using~\Cref{thm:mainthm}, we have that the hybrids $\hybrid_3$ and $\hybrid_4$ are identically distributed. Note that the combined registers $(\regR,\regA)$ in the above hybrid will take the role of $\regA$ in~\Cref{thm:mainthm}. \\

\noindent We now design $\bm{\alice'}=(\alice,\bob_1,\ldots,\bob_t)$:
\begin{itemize}
    \item $\alice:$ Upon receiving $\rho^{\otimes t}_{\regA[1],\ldots,\regA[t]}$, first execute $\sims\left( \rho^{\otimes t}_{\regA[1],\ldots,\regA[t]} \right)$ to obtain $\sigma$. It then executes $\alice(\sigma)$ to obtain $(t+1)$ registers $\regZ_{\bob'_1},\ldots,\regZ_{\bob'_{t+1}}$. It sends $\regZ_{\bob'_i}$ to $\bob_i$. 
    \item $\bob_i$: upon receiving $x_i$, it runs $\bob'_i$ on $(\regZ_{\bob'_i},x_i)$ to obtain $y_i$. Output $y_i$. 
\end{itemize}
\noindent The probability that $\bm{\alice'}$ outputs $(y_1,\ldots,y_{t+1})=(f(x_1),\ldots,f(x_{t+1}))$ is negligibly close to $p$ and hence, non-negligible. This contradicts the iid multi-copy security of $(\copyprotect,\Eval)$. 
\end{proof}

Instantiating $(\copyprotect,\Eval)$ using the scheme by~\cite{CG24b}, we obtain the following corollary. 

\begin{corollary}
Assuming post-quantum sub-exponentially secure indistinguishability obfuscation and learning with errors, there exists identical-copy secure copy-protection schemes for digital signatures and pseudorandom functions. 
\end{corollary}
\fi

\section{$t$-copy Pseudorandom States}
\label{sec:prs}

We will begin by remarking on a useful property of Haar random states. \pnote{is the reference correct below? I think it was done in BS19 and AGQY if I'm not mistaken. JLS doesn't consider binary phase}\enote{This is not binary phase. Binary phase is the distinct case.}
\begin{lemma}[See proof of Lemma 1~\cite{JLS18}]\label{lem:haartypelemma}
    Let $n,t\in \N$. For $\vec{r}=(r_1,\dots,r_t)\in (\{0,1\}^{n})^t$, define 
    $$\ket{perm_{\vec{r}}} \propto \sum_{\pi \in Sym(t)} \bigotimes_{j=1}^t \ket{r_{\pi(t)}}$$
    Then
    $$\TD\left(\E_{\ket{\phi}\gets \Haar(\{0,1\}^n)}\left[\ketbra{\phi}^{\otimes t}\right],\E_{\vec{r}\gets \{0,1\}^{n\cdot t}}\left[\ketbra{perm_{\vec{r}}}\right]\right) \leq O\left(\
\frac{t^2}{2^n}\right)$$
\end{lemma}

Let $G$ be a psuedorandom state generator. We will assume without loss of generality that $G(k)$ acts as follows
\begin{enumerate}
    \item Apply a unitary $U_G$ to the state $\ket{k}\ket{0}$, producing a state $\ket{\phi_k}_{AB}\ket{0}_C$
    \item Output $\Tr_B(\ketbra{\phi}_{AB})$.
\end{enumerate}
We will say that $G$ produces $\ell_j$ bits of junk, where the register $B$ is over $\mathcal{H}(\{0,1\}^{\ell_j})$.

\begin{theorem}\label{thm:prsthm}
    Let $\{\ket{\phi_k}\}$ be a single-copy pseudorandom state generator with keys of length $\ell_k(\secparam)$ over states of length $\ell_n(\secparam)$ producing $\ell_j(\secparam)$ bits of junk.
    
    Let $t(\secparam),\ell'(\secparam)$ be any polynomials such that $\ell'=\omega(\log\secparam)$. Let $\{f_{1,k}:\{0,1\}^{\ell'} \to [t+1]\},\{f_{2,k}:\{0,1\}^{\ell'}\to \{0,1\}^{\ell_j}\},\{f_{3,k}:\{0,1\}^{\ell'}\to \{0,1\}^{\ell_j}\},\{f_{4,k}:\{0,1\}^{\ell'}\to \{0,1\}^{\ell_k}\}$ be four $2t$-wise independent hash function famillies with keys of length $\ell_{k_{f_1}},\dots,\ell_{k_{f_4}}$ respectively.

    Then there exists a $t$-time pseudorandom state generator with keys of length $\ell'_k = \ell_{k_{f_1}}+\ell_{k_{f_2}}+\ell_{k_{f_3}} + \ell_{k_{f_4}}$ over states of length $\ell'_n=\ell'+\ell_n + \ell_j$.
\end{theorem}

Applying~\Cref{thm:efftwise} and setting $\ell' = \ell_k$ gives
\begin{corollary}
    Let $\{\ket{\phi_k}\}$ be a single-copy pseudorandom state generator with keys of length $\ell_k(\secparam)$ over states of length $\ell_n(\secparam)$ producing $\ell_j(\secparam)$ bits of junk.
    
    Let $t$ be any polynomially bounded function. Then there exists a $t$-time pseudorandom state generator with keys of length $O(t\cdot (\ell_k+\ell_j))$ over states of length $\ell_k + \ell_n + \ell_j$.
\end{corollary}

\begin{proof}

The construction will be exactly~\Cref{cons:mainthm} where $f_1,\dots,f_4$ are instantiated by the $2t$-wise independent functions. Formally, $\wt{G}(k_{f_1},\dots,k_{f_2})$ will output
$$\frac{1}{\sqrt{2^{\ell'}}}\sum_{i \in \{0,1\}^{\ell'}}\omega_{t+1}^{f_{1,k_{f_1}}(i)}(I_{\regA} \otimes X^{f_{2,k_{f_2}}(i)}_{\regB}Z^{f_{3,k_{f_3}}(i)}_{\regB} \otimes I_{\regC}) \ket{\phi_{f_{4,k_{f_4}}(i)}}_{\regB}\otimes \ket{i}_{\regC}$$

\noindent We show that $\widetilde{G}$ is a secure pseudorandom generator by a hybrid argument. Let $\adv$ be a QPT adversary that takes as input $t$ copies of a state and it needs to distinguish whether this state is a PRS state or is it Haar random. Define $p_i$ the probability that $\adv$ outputs $1$ in $\hybrid_i$.\\

\noindent $\hybrid_1$: $\adv$ receives as input $\widetilde{G}\left( \widetilde{k} \right)^{\otimes t}$, where $\widetilde{k} \xleftarrow{\$} \{0,1\}^{\secparam'}$. \\

\noindent $\hybrid_2$: Sample $f_1,\dots,f_4$ uniformly random functions. $\adv$ receives as input
$$\frac{1}{\sqrt{2^{\ell'}}}\sum_{i \in \{0,1\}^{\ell'}}\omega_{t+1}^{f_{1}(i)}(I_{\regA} \otimes X^{f_{2}(i)}_{\regB}Z^{f_{3}(i)}_{\regB} \otimes I_{\regC}) \ket{\phi_{f_{4}(i)}}_{\regB}\otimes \ket{i}_{\regC}$$
Since $f_1,\dots,f_4$ are replacing $2t$-designs, by~\Cref{lem:twise} $p_2=p_1$.\\

\noindent $\hybrid_3$: Let $\Sim_{PRS}$ be the algorithm from~\Cref{cons:sim} for the family $\{\ket{\phi_k\}}$. Sample $k_1,\dots,k_t$. $\adv$ receives as input 
$$\Sim(\Tr_{\regB}(\ketbra{\phi_{k_1}}_{\regA \regB})\otimes \dots \otimes \Tr_{\regB}(\ketbra{\phi_{k_t}}_{\regA \regB}))$$
By~\Cref{thm:mainthm}, $\abs{p_2-p_1}\leq \negl(\secparam)$.\\

\noindent $\hybrid_4$: Sample $\ket{\phi_1},\dots,\ket{\phi_k}$ uniformly random states over $\mathcal{H}(\{0,1\}^{\ell_n})$. $\adv$ receives as input 
$$\Sim(\ketbra{\phi_1}\otimes \dots \otimes \ketbra{\phi_t})$$
Note that in $\hybrid_3$, the entire game gets access to exactly one copy of each $\Tr_{\regB}(\ketbra{\phi_{k_i}}_{\regA\regB})$. By appling single-copy pseudorandom state generator security security for each $i$, we get $\abs{p_4-p_3}\leq \negl(\secparam)$

\noindent $\hybrid_5$: Sample $r_1,\dots,r_k\gets \{0,1\}^{\ell_n}$. $\adv$ receives as input
$$\Sim(\ket{r_1}\otimes \dots \otimes \ket{r_t})$$
This follows immediately from the fact that the mixed state representing one copy of a Haar random state is exactly the maximally mixed state. And so $p_5=p_4$.\\

\noindent $\hybrid_6$: Sample $r_1',\dots,r_k' \gets \{0,1\}^{\ell'_n}$. $\adv$ receives as input
$$\propto\sum_{\pi \in Sym(\ell'_n)}\bigotimes_{j=1}^n\ket{r_{\pi(j)}'}$$
This state is exactly the state $\Sim$ produces on input $\ket{r_1}\dots\ket{r_t}$ for random $r_1,\dots,r_t$. Thus, $p_6=p_5$.\\

\noindent $\hybrid_7$: Sample $\ket{\psi}\gets \mathcal{H}(\{0,1\}^{\ell'_n})$. $\adv$ receives as input $\ket{\psi}^{\otimes t}$.
By~\Cref{lem:haartypelemma}, the state $\adv$ receives in $\hybrid_6$ and $\hybrid_7$ are negligibly close in trace distance. And so the probability that $\adv$ outputs $1$ in both games will be negligibly close. That is, $\abs{p_7-p_6}\leq \negl(\secparam)$.\\

Combining all these hybrids together, we get $\abs{p_7-p_1}\leq \negl(\secparam)$ and so $\wt{G}$ is a $t$-copy secure pseudorandom state generator.

\end{proof}

\section{Simulating non-adaptive queries to a family of unitaries}
\subsection{Notation}

\pnote{I think we need consistent notation for registers. Some places we use $K$, the other places we use ${\bf X}$ etc. I'll try to make it consistent below.}

\newcommand{\V}{\mathcal{V}}

\begin{definition}
    For a set $S$, we define $\mathcal{H}(S)$ to be the Hilbert space of dimension $|S|$ generated by $\ket{s}$ for $s\in S$.
\end{definition}

\begin{definition}
    Let $\V=\{V_i\}_{i\in [N]}$ be some family of isometries. Let $\regIn$ be an input register for $\V$, and let $\regOut$ be an output register. Moreover, let $\regK$ be a register on a Hilbert space of dimension $N$. Define the isometry $\Apply^\V_{\regK,\regIn}$ \pnote{why is the superscript ${\cal U}$ and not ${\cal V}$?}
    $$\Apply^\V_{\regK, \regIn} (\ket{k}_{\regK}\otimes \ket{x}_{\regIn}) \mapsto \ket{k}_{\regK} \otimes (V_k \ket{x})_{\regOut}$$
    \pnote{why are the registers $\regIn$ and $\regOut$ different? is it because $V_k$ is an isometry?}\enote{Yes}
\end{definition}

\begin{definition}
    For a function $f:\{0,1\}^n\to [q]$, define the unitary $S^f$ to be the map acting over $\mathcal{H}(\{0,1\}^{n})$ by
    $$S^f \ket{x} \mapsto \omega_q^{f(x)} \ket{x}$$
\end{definition}

\begin{definition}
    For any $\ell,t$, we define a projector $\Pi_{{\sf dist}}^{\ell,t}$ over $\mathcal{H}(\{0,1\}^{\ell})^{\otimes t}$ by
    $${\sf Im}(\Pi_{{\sf dist}}^{\ell,t}) = {\sf Span}(\{\ket{x_1,\dots,x_t}:x_1\neq \dots \neq x_t\}),$$
    where: ${\sf Im}(\Pi)$, for a projector $\Pi$, is defined to the set of all $\ket{u}$ such that $\Pi \ket{u} = \ket{u}$.
\end{definition}

\begin{definition}
    For a set $S$ and any $n\in \N$, we define
    $${\sf MS}^{S,n} = \{S' \subseteq_{{\sf ms}} S : |S'|= n\}$$
    to be the set of multisets containing at $n$ elements from $S$. We say that $A \subseteq_{{\sf ms}} B$ if $A \subseteq B$ and $A$ is a multiset.   

    Similarly, define
    $${\sf MS}^{S,\leq n} = \{S' \subseteq_{{\sf ms}} S : |S'|\leq n\}$$
    be the set of multisets containing at most $n$ elements from $S$.

    We will identify multisets with sorted lists of elements, possibly containing duplicates.
\end{definition}

\begin{definition}
    For registers $\regK,\regR,\regK',\regR'$ where $\regK,\regK'$ are over $\mathcal{H}(\{0,1\}^{\secparam})$, we define the unitary $\Select$ which swaps registers $\regR,\regR'$ if and only if the values in registers $\regK,\regK'$ are the same. That is,
    $$\Select\ket{k}_{\regK}\ket{x}_{\regR}\ket{k'}_{\regK'}\ket{x'}_{\regR'} \mapsto \begin{cases}
        \ket{k}_{\regK}\ket{x'}_{\regR}\ket{k'}_{\regK'}\ket{x}_{\regR'} & k=k'\\
        \ket{k}_{\regK}\ket{x}_{\regR}\ket{k'}_{\regK'}\ket{x'}_{\regR'} & k\neq k'
    \end{cases}$$
\end{definition}

\begin{definition}
    An oracle $\mathcal{O}$ is defined by an isometry acting over an input register $\regIn$ and an internal register $\regSt$. An oracle algorithm $\mathcal{A}^{\mathcal{O}}$ is a sequence of isometries $\A^1,\dots,\A^t$ acting on registers $\regX_1 \otimes \regIn_1,\dots,\regX_t \otimes \regIn_t$ with output register $\regY$. On any input state $\ket{\phi}_{\regX_1,\regIn_1,\regSt}$, the evaluation of $\mathcal{A}$ on $\ket{\phi}$ is the state
    $$(\A^{\mathcal{O}}\ket{\phi})_{\regY,\regSt} = \A^t_{\regX_t, \regIn_t}\cdot \mathcal{O}_{\regIn_{t-1},\regSt}\cdot \A^{t-1}_{\regX_{t-1},\regIn_{t-1}}\cdots \mathcal{O}_{\regIn_1,\regSt}\cdot \A^1_{\regX_1,\regIn_1}\ket{\phi}_{\regX_1,\regIn_1,\regSt}$$
\end{definition}

\subsection{Main Theorem - Unitary setting}

\begin{theorem}[Main Theorem For Unitary Setting]\label{thm:unisimmain}
    Let $\U=\{U_k\}_{k\in \{0,1\}^\secparam}$ be a collection of unitaries acting on some Hilbert space $\mathcal{H}^{\U}$. Let $t\in \N$.

    Let $\regIn=\regIn_1,\dots,\regIn_t$ be a register over $(\mathcal{H}^U)^{\otimes t}$. Let $\regK=\regK_1,\dots,\regK_t$ be a register over $(\mathcal{H}(\{0,1\}^{\secparam}))^{\otimes t}$.

    There exists a CPTP map $\Sim^t$ such that the following holds: Let $\rho_{\regIn,\regK}$ be any state such that $\Tr(((\Pi_{dist}^{\secparam,t})_{\regK} \otimes I_{\regIn})\rho_{\regK,\regIn}) = 1$. For any $f:\{0,1\}^{\secparam} \to [2t]$, define:  $$\sigma^f_{\regK,\regIn} = \left( \bigotimes_{i=1}^t S^{f}_{\regK_i}\cdot \Apply^{\U}_{\regIn_i}\right) \cdot \rho_{\regK,\regIn} \cdot \left(\bigotimes_{i=1}^t S^{f}_{\regK_i}\cdot \Apply^{\U}_{\regIn_i}\right)^\dagger$$

    Then 
    $$\E_{f}\left[ \sigma^f_{\regK,\regIn} \right] = \Sim^t(\rho_{\regK,\regIn})$$

    Furthermore, $\Sim^t$ can be efficiently implemented, in time $\poly(\secparam,t)$, by an algorithm of the following form: first it chooses distinct classical keys $k_1,\dots,k_t$, then it queries each $U_{k_i}$ exactly once. \pnote{this sentence is a bit ambiguous since it is not clear what making a query to the set $\U$ is. How about the following? Furthermore, $\Sim^t$ can be efficiently implemented, in time $\poly(\secparam,t)$, by an algorithm that has oracle access to $U_{k_1},\ldots,U_{k_t}$, where $k_i \xleftarrow{\$} \{0,1\}^{\secparam}$.}\enote{I agree that it is ambiguous, but your phrasing leaves out the fact that you can only query each $U_k$ once. Also, the simulator gets to choose $k_1,\dots,k_t$. Here's another attempt}
\end{theorem}

Concretely, $\Sim^t$ will operate on ancilla registers $\regSt$ over $\mathcal{H}({\sf MS}^{\{0,1\}^\secparam,t})$ and $\regR = \regR_1,\dots,\regR_t$ over $(\mathcal{H}^U)^{\otimes t}$. $\Sim^t(\rho_{\regK,\regIn})$ will be defined as follows:
\begin{enumerate}
    \item Initialize register $\regSt$ to $\ket{\emptyset}$ and register $\regR=\regR_1,\dots,\regR_t$ to $\ket{\vec{0}}$.
    \item Apply $({\sf Cntrl}{\tiny -}\uplus)_{\regK,\regSt}$ \pnote{is this supposed to be $\regSt$?} defined by
    $$({\sf Cntrl}{\tiny -}\uplus)_{\regK,\regSt}\ket{\vec{k}}_{\regK}\ket{S}_{\regSt} \mapsto \ket{\vec{k}}_{\regK}\ket{S\uplus\{k_1,\dots,k_t\}}_{\regSt}$$
    \item Measure the register $\regSt$ in the standard basis. This produces a sorted list of classical keys $(k_1,\dots,k_t)$ on registers $\regSt_1,\dots,\regSt_t$.
    \item For each $i,j\in [t]$, run $\Select_{\regK_i,\regIn_i,\regSt_j,\regR_j}$. \pnote{is this supposed to be $\regSt_i,\regR_i$ instead of $\regSt,\regR$? I guess you are iterating over all $\regSt_j,\regR_j$, right? }\enote{You are correct}
    \item For each $i\in [t]$, apply $U_{k_i}$ to register  $\regR_i$.
    \item For each $i,j\in [t]$, run $\Select_{\regK_i,\regIn_i,\regSt_j,\regR_j}$ again. \pnote{is this supposed to be $\regSt_i,\regR_i$ instead of $\regSt,\regR$?}\enote{Resolved}
    \item Output the registers $\regK,\regIn$. 
\end{enumerate}

Note that when $N=2$ and $U_1$ is the identity, then $\Apply^\U_{\regK,\regIn}$ implements controlled access to $U_2$. That is, $\Apply^{\U}_{\regK,\regIn}=\ketbra{1}{1}_{\regK} \otimes I_{\regIn} + \ketbra{2}{2}_{\regK} \otimes (U_2)_{\regIn}$.   ~\cite{tang2025controlledunitarieshelpful} showed that for any fixed unitary $U$, $t$ queries to controlled access to $\omega_{2t}^\theta U$, for a random $\theta$, can be simulated by $t$ queries to $U$. This can be generalized, showing that for any family $\U$ and for a random $f$, $S^f_{\regK} \cdot \Apply^\U_{\regK,\regIn}$ can be simulated using only forward queries to $\U$. The key idea behind our proof is that $\Sim^t$ implements this simulator for the specific case of non-adaptive queries to (a superposition) of distinct keys.

\ifnum\llncs=0
\subsection{Simulating adaptive queries to families of isometries}

We will begin by developing a simulator $\Sim_{\iso}^{s,t}$ which will emulate a family of isometries $\V=\{V_k\}$ up to $s$ number of $t$-parallel queries.

\begin{theorem}\label{thm:isosim}
    Let $\V=\{V_k\}_{k\in [N]}$ be an arbitrary collection of unitaries. Let $s,t\in \N$. Define $q = 2st$. Let $\regK=\regK_1,\dots,\regK_t$ be registers over $\mathcal{H}(\{0,1\}^{\secparam})$. Let $\regIn=\regIn_1,\dots,\regIn_t$ and $\regOut=\regOut_1,\dots,\regOut_t$ be respectively $t$ input and output registers for $\V$.

    Let $\regSt'$ be a register over $\mathcal{H}([N]^q)$. We define an oracle $\Sim_{\iso}^{s,t}$ acting on registers $\regK,\regIn,\regOut$ and an internal state $\regSt'$ as follows. \pnote{instead of $y_j$, should it be $x_j$? what about $T + \type(\vec{k})$? is it component-wise addition?}
    \begin{equation}
        \Sim_{\iso}^{s,t} \ket{\vec{k}}_{\regK}\ket{\vec{x}}_{\regIn}\ket{T}_{\regSt'} \mapsto \ket{\vec{k}}_{\regK} \otimes \left(\bigotimes_{j=1}^t (V_{\regK_j} \ket{x_j})_{\regOut_j}\right) \otimes \ket{T + \type(\vec{k})}_{\regSt'}
    \end{equation}
    where $T+\type(\vec{k})$ represents component-wise addition.

    Let $\A^{(\cdot)}$ be any oracle algorithm making at most $s$ queries \pnote{is it parallel queries?} to its oracle with input register $\regX$ and output register $\regY$. Let $\mathcal{F}$ be the set of all functions $[N]\to [q]$. 
    
    Define the states \pnote{the oracle access below needs to be explicitly defined. For instance, in the oracle access to $\Sim$, the register $\regSt$ is not revealed to $\A$, right?}
    \begin{equation}
    \begin{split}
        \ket{\phi_f}_{\regY} = \A^{\left((S^f \otimes I)\cdot \Apply^\U \right)^{\otimes t}}\ket{0}_{\regX}\\
        \ket{\psi}_{\regY,\regSt'} = \A^{\Sim^{s,t}_{\iso}}\ket{0}_{\regX}\ket{\vec{0}}_{\regSt'}
    \end{split}
    \end{equation}

    Then
    \begin{equation}
        \E_{f\gets \mathcal{F}}\left[ \ketbra{\phi_f} \right] = \Tr_{\regSt'}\left(\ketbra{\psi}_{\regY,\regSt'} \right)
    \end{equation}
\pnote{read till here}
\end{theorem}

\begin{proof}
    Recall that $\regSt'$ is a register over $\mathcal{H}([N]^q)$. We will identify $[N]^q$ with $\mathcal{F}$, the functions from $[N]\to[q]$. Let $\ket{P}_{\regY, \regSt'} = \frac{1}{\sqrt{|\mathcal{F}|}}\sum_{f\in \mathcal{F}} \ket{\phi_f}_{\regY}\ket{f}_{\regSt'}$ be a purification of
    \begin{equation}\label{eq:firsteq}
        \E_{f\gets \mathcal{F}}\left[ \ketbra{\phi_f} \right] = \Tr_{\regSt'}(\ketbra{P}_{\regY,\regSt'})
    \end{equation}

    Let us define an isometry $\PureV$ to act as follows
    $$\PureV\left(\ket{\vec{k}}_{\regK}\ket{\vec{x}}_{\regIn}\ket{f}_{\regSt'} \right) \mapsto \ket{\vec{k}}_\regK\otimes \left(\bigotimes_{j=1}^t(\omega_{q}^{f(\regK_j)} V_{\regK_j}\ket{y_j})_{\regOut_j}\right)\otimes \ket{f}_{\regSt'}$$

    In particular, \pnote{the normalization factor missing below}
    $$\ket{P} = \A^{\PureV}\ket{0}_{\regX} \otimes \left(\frac{1}{\sqrt{q^N}} \sum_{f\in [q]^N} \ket{f}_{\regSt'} \right)$$

    The proof immediately follows from the fact that $\Sim^{s,t}$ is exactly $\PureV$ conjugated by $QFT_q^{\otimes t}$ on register $\regSt$. Observe that for any $f,\vec{k}$,
    \begin{equation}
        \prod_{j=1}^s \omega_q^{f(\regK_j)} = \prod_{r=1}^N \omega_q^{f(r)\cdot \type(\vec{k})_r} = \omega_q^{f\cdot \type(\vec{k})}
    \end{equation}
    and so
    \begin{equation}
    \begin{split}
        \PureV_{\regK,\regIn,\regSt}\ket{\vec{k}}_{\regK}\ket{\vec{x}}_{\regIn}\ket{f}_{\regSt'} = \ket{\vec{k}}_\regK\otimes \omega_q^{f\cdot \type(\vec{k})}\left(\bigotimes_{j=1}^t( V_{\regK_j}\ket{y_j})_{\regOut_j}\right)\otimes \ket{f}_{\regSt'}\\
        =\ket{\vec{k}}_\regK\otimes \left(\bigotimes_{j=1}^t( V_{\regK_j}\ket{y_j})_{\regOut_j}\right)\otimes \omega_q^{f\cdot \type(\vec{k})}\ket{f}_{\regSt'}
    \end{split}
    \end{equation}

    We can then explicitly compute $(QFT_q^{\otimes t})_{\regSt}^{\dagger} \cdot \PureV_{\regK,\regIn,\regSt'} \cdot (QFT_q^{\otimes t})_{\regSt}$.

    Let $\ket{\vec{k}}_{\regK} \ket{\vec{x}}_{\regIn} \ket{T}_{\regSt}$ be any basis state. We then evaluate
    \begin{equation}
        \begin{split}
             (QFT_q^{\otimes t})_{\regSt}^{\dagger} \cdot \PureV_{\regK,\regIn,\regSt'} \cdot (QFT_q^{\otimes t})_{\regSt}\ket{\vec{k}}_{\regK} \ket{\vec{x}}_{\regIn} \ket{T}_{\regSt}\\
             = (QFT_q^{\otimes t})_{\regSt}^{\dagger} \cdot \PureV_{\regK,\regIn,\regSt'} \ket{\vec{k}}_{\regK}\ket{\vec{x}}_{\regIn}\sum_{f\in [N]^q}\frac{1}{\sqrt{q^T}}\omega_q^{f\cdot T}\ket{f}_{\regSt'}\\
            = (QFT_q^{\otimes t})_{\regSt}^\dagger \left(\ket{\vec{k}}_{\regK}\otimes \left(\bigotimes(V_{\regK_j}\ket{y_j})_{\regOut_j}\right)\otimes \sum_{f\in [N]^q}\frac{1}{\sqrt{q^T}}\omega_q^{f\cdot \type(\vec{k})}\cdot \omega_q^{f\cdot T}\ket{f}_{\regSt'}\right)\\
             =  (QFT_q^{\otimes t})_{\regSt}^\dagger \left(\ket{\vec{k}}_{\regK}\otimes \left(\bigotimes(V_{\regK_j}\ket{y_j})_{\regOut_j}\right)\otimes \sum_{f\in [N]^q}\frac{1}{\sqrt{q^T}}\omega_q^{f\cdot (\type(\vec{k})+ T)}\ket{f}_{\regSt'}\right)\\
             =  \ket{\vec{k}}_{\regK}\otimes \left(\bigotimes(V_{\regK_j}\ket{y_j})_{\regOut_j}\right)\otimes \ket{T+\type(\vec{k})}_{\regSt'}\\
            = \Sim^{s,t}_{\regK,\regIn,\regSt}\ket{\vec{k}}_{\regK} \ket{\vec{x}}_{\regIn} \ket{T}_{\regSt'}
        \end{split}
    \end{equation}

    Since the Fourier transform only acts on the state register, we can telescope terms to get
    \begin{equation}\label{eq:maineq}
        \begin{split}
            \ket{\psi}_{\regY,\regSt} = \A^{\Sim^{s,t}}\ket{0}_X\ket{\vec{0}}_{\regSt'}\\
            =(I_Y \otimes (QFT_q^{\otimes t})^{\dagger}_{\regSt'})\A^{\PureV} (I_{\regX} \otimes (QFT_q^{\otimes t})_{\regSt'}) \ket{0}_X\ket{\vec{0}}_{\regSt'}\\
            =(I_{\regY} \otimes (QFT_q^{\otimes t})^{\dagger}_{\regSt'})\A^{\PureV}  \ket{0}_X\otimes \sum_{f\in [N]^q}\ket{f}_{\regSt'}\\
            =(I_{\regY} \otimes (QFT_q^{\otimes t})^{\dagger}_{\regSt'}) \ket{P}
        \end{split}
    \end{equation}
    since $(QFT^{\otimes t})^{\dagger}_{\regSt'}$ only acts on the state register, we get
    \begin{equation}
        \Tr_{\regSt}(\ketbra{\psi}) =\Tr_{\regSt}(\ketbra{P}) = \E_{f\gets \mathcal{F}}[\ketbra{\phi_f}]
    \end{equation}
\end{proof}




A major downside of this simulator is that its internal state grows with $N$, which may be exponential. Here, we take a page from~\cite{zhandry2019record}, and observe that the sum of all values in the internal state is bounded by $s\cdot t$. Thus, it is sufficient to instead store a list of all values contained in the state, that is a multiset.

\begin{theorem}\label{thm:isosimeff}
    Let $\V,s,t,\regK,\regIn,\regOut,\Sim^{s,t}_{\iso}$ be as in~\Cref{thm:isosim}.

    Define $\regSt$ to be a register over $\mathcal{H}\left({\sf MS}^{[N],\leq st} \right)$.     
    
    We define an isometry $\ESim^{s,t}$ acting on registers $\regK,\regIn,\regOut,\regSt$ as follows.
    \begin{equation}
        \ESim^{s,t} \ket{\vec{k}}_{\regK}\ket{\vec{x}}_{\regIn}\ket{S}_{\regSt} \mapsto \ket{\vec{k}}_{\regK} \otimes \left(\bigotimes_{j=1}^t (V_{\regK_j} \ket{y_j})_{\regOut_j}\right) \otimes \ket{S \uplus \{\regK_1,\dots,\regK_t\}}_{\regSt}
    \end{equation}

    Then for all oracle algorithms $\A^{(\cdot)}$ making at most $s$ queries to its oracle with input $X$ and output register $Y$,
    let 
    $$\ket{\psi}_{Y,\regSt} = \A^{\Sim^{s,t}}\ket{0}_X\ket{\emptyset}_{\regSt'}$$
    $$\ket{\psi'}_{Y,\regSt'} = \A^{\ESim^{s,t}}\ket{0}_X\ket{\emptyset}_{\regSt}$$
    Then
    $$\Tr_{\regSt}(\ketbra{\psi}_{Y,\regSt'})=\Tr_{\regSt}(\ketbra{\psi'}_{Y,\regSt'})$$
\end{theorem}

\begin{proof}
    Given a multiset $S\in MS^{[N],\leq st}$, we can define a vector $\vec{v}^S$ by $\vec{v}^S_i = $the number of times $i$ appears in $S$. Let $\expand$ be the isometry mapping $\mathcal{H}({\sf MS}^{[N],\leq q}) \to \mathcal{H}([N]^{q})$ defined by
    $$\expand\ket{S} \mapsto \ket{\vec{v}^S}$$

    Define $\Pi_{\leq r}$ to be the projector with ${\sf Im}(\Pi_{\leq r})= {\sf Span}\{\vec{v}\in [N]^q:\sum_{i\in [N]} v_i \leq r\}$. We have that for all $r\leq st$, ${\sf Im}(\Pi_{\leq r}) \subseteq {\sf Im}(\expand) = {\sf Im}(\Pi_{\leq st})$ \pnote{should it be ${\sf Im}\left(\Pi_{\leq r}\right) \subseteq {\sf Im}(\expand) = {\sf Im}\left(\Pi_{\leq st}\right)$}. Each query to $\Sim^{s,t}$ maps a state in $I\otimes \Pi_{\leq r}$ to a state in $I\otimes \Pi_{\leq r+t}$, and so by induction after each query to $\Sim^{s,t}$, we have that the resulting state is contained in $I \otimes \Img(\expand)$.

    We will then see that for any input \pnote{is it ok if I format the equations below such that they are aligned from the left? I'll not do any edits if you prefer to keep it this way.}\enote{Of course its fine}
    \begin{equation}
        \begin{split}
            (I_{\regK,\regOut}\otimes \expand^{\dagger}_{\regSt'}) \Sim^{s,t} (I_{\regK,\regIn} \otimes  \expand_{\regSt})\ket{\vec{k}}_{\regK}\ket{\vec{x}}_{\regIn}\ket{S}_{\regSt}\\
            =(I_{\regK,\regOut}\otimes \expand^{\dagger}_{\regSt'}) \Sim^{s,t}\ket{\vec{k}}_{\regK} \ket{\vec{x}}_{\regIn}\ket{\vec{v}^S}_{\regSt'}\\
            =(I_{\regK,\regOut}\otimes \expand^{\dagger}_{\regSt'}) \ket{\vec{k}}_{\regK}\otimes \left(\bigotimes_{j=1}^t V_{\regK_j}\ket{x_j}\right)_{\regOut} \otimes \ket{\vec{v}^S + \type(\vec{k})}_{\regSt'}\\
            =\ket{\vec{k}}_{\regK}\otimes \left(\bigotimes_{j=1}^t V_{\regK_j}\ket{x_j}\right)_{\regOut} \otimes \ket{S\uplus \{\regK_1,\dots,\regK_t\}}_{\regSt}\\
            = \ESim^{s,t} \ket{\vec{x}}_{\regIn}\ket{S}_{\regSt}
        \end{split}
    \end{equation}

    By telescoping, we then get
    $$\ket{\psi'} = (I\otimes \expand^{\dagger}_{\regSt'})\ket{\phi}$$
    and since $\expand$ only acts on the state register the theorem follows.
\end{proof}

\begin{figure}
    \centering
    \begin{quantikz}[transparent]
    \lstick{$\rho_{\regSt'}$}\qw & \gate{\uplus} &     \qw    & \qw & \qw & \qw\\
    \lstick{$\rho_{K}$}\qw & \ctrl{-1} & \gate[wires=7,nwires={5}]{\Select} & \qw & \gate[wires=7,nwires={5}]{\Select} & \qw\\
    \lstick{$\rho_{In}$}\qw & \qw &     \qw    & \qw & \qw & \qw\\
    \lstick{$\ket{1}_{\regK_1}$}\qw & \qw &     \qw    & \qw & \qw & \qw\\
    \lstick{$\ket{0}_{R'_1}$}\qw & \qw & \qw & \gate{V_1} & \qw & \qw\\
        &     &     &   \myvdots   &   & \\
    \lstick{$\ket{N}_{\regK_N}$}\qw & \qw &     \qw    & \qw & \qw & \qw\\
    \lstick{$\ket{0}_{R'_N}$}\qw & \qw & \qw & \gate{V_N} & \qw & \qw\\
    \end{quantikz}
    \caption{The efficient implementation of the simulator for general isometries. Here the $\uplus$ gate represents adding the string on register $\regK$ to the multiset stored in $\regSt'$.}
    \label{fig:onetimeuni}
\end{figure}
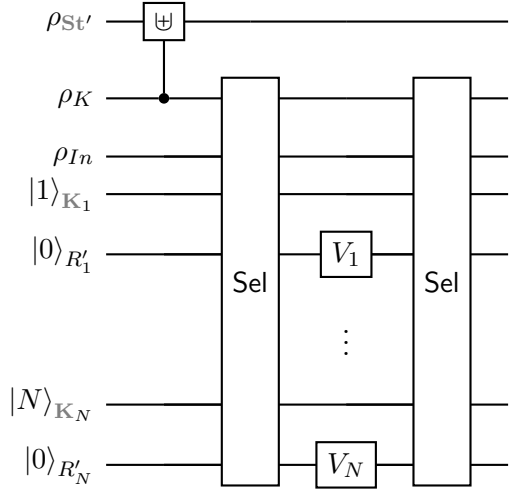

Note that when $N$ is polynomial, $\ESim^{s,t}$ already has an efficient implementation. We show the circuit for this implementation in~\Cref{fig:onetimeuni} for the $t=1$ case. Larger values of $t$ can be simulated by calling $\ESim^{st,1}$ $t$ times for each parallel query. Proof follows by explicit computation, but since we do not use this theorem we omit the details

\subsection{Proof of correctness of our simulator}
Instead of dealing with $\Sim^t$ directly, we will work with its purification $\Sim_{{\sf pure}}^t$ acting on registers $\regK,\regIn,\regSt$ defined as follows: \pnote{are we using $\regSt$ or $\regSt'$? in some places below, $\regSt$ is used. }
\begin{enumerate}
    \item Initialize register $\regR$ to $\ket{\vec{0}}$.
    \item Apply $({\sf Cntrl}{\tiny -}\uplus)_{\regK,\regSt}$ defined by
    $$({\sf Cntrl}{\tiny -}\uplus)_{\regK,\regSt}\ket{\vec{k}}\ket{S}\mapsto \ket{\vec{k}}\ket{S\uplus\{k_1,\dots,k_t\}}$$
    \item For each $i,j\in [t]$, run $\Select_{\regK_i,\regIn_i,\regSt_j,\regR_j}$. \pnote{$\regSt$ or $\regSt'$?}
    \item For each $i\in [t]$, run $\Apply_{\regSt_i,\regR_i}^{\mathcal{U}}$. \pnote{the $\Apply$ isometry is defined to be parameterized by a set of unitaries. Which set of isometries are associated with $\Apply$?}
    \item For each $i,j\in [t]$, run $\Select_{\regK_i,\regIn_i,\regSt_j,\regR_j}$ again. \pnote{$\regSt$ or $\regSt'$?}
    \item Output registers $\regK,\regIn,\regSt$.
\end{enumerate}
In particular, instead of measuring $\regSt$, it will coherently apply $U_{k_i}$ on the correct registers.

\begin{lemma}\label{lem:pureuni}
    Let $\U=\{U_k\}_{k\in \{0,1\}^{\secparam}}$ be a family of unitaries. Let $t\in \N$. Let $\regK=\regK_1,\dots,\regK_t,\regIn=\regIn_1,\dots,\regIn_t$ be registers over $\mathcal{H}(\{0,1\}^{\secparam})^{\otimes t}$ and $(\mathcal{H}^\U)^{\otimes t}$ respectively \pnote{maybe we should use a different notation compared to $(\mathcal{H}^\U)$?}. For all states $\rho_{\regK,\regIn}$ such that ${\sf Tr}((\Pi_{{\sf dist,\regK}}^{\secparam,t}\otimes I_{\regIn})\rho_{\regK,\regIn}) = 1$ \pnote{on which registers would $\Pi_{{\sf dist}}$ act upon? the superscripts should be $(\secparam,t)$?},
    $$\Sim^t_{{\sf pure}}(\rho \otimes \ketbra{\emptyset})  = \ESim^{s,t}(\rho \otimes \ketbra{\emptyset})$$
\end{lemma}

\begin{proof}
    This follows by simple computation. Let $\ket{\phi_i}$ be the state after step $i$ when running $\Sim_{pure}^t$ with initial state $\ket{\phi_1} = \ket{\vec{k}}_{\regK}\ket{\vec{x}}_{\regIn}\ket{\emptyset}_{\regSt'}\ket{\vec{0}}_{\regR}$. Let $i_j$ be the index of the $j$th largest element in $\vec{k}$, which is unique by assumption. 

    \begin{equation}
        \begin{split}
            \ket{\phi_1} = \ket{\vec{k}}_{\regK}\ket{\vec{x}}_{\regIn}\ket{\emptyset}_{\regSt}\\
            \ket{\phi_2} = \ket{\vec{k}}_{\regK}\ket{\vec{x}}_{\regIn}\ket{Sort(\vec{k})}_{\regSt}\\
            \ket{\phi_3} = \ket{\vec{k}}_{\regK}\ket{\vec{x}}_{\regIn}\ket{k_{i_1},\dots,k_{i_t}}_{\regSt}\ket{\vec{0}}_{\mathbf{R'}}\\
            \ket{\phi_4} = \ket{\vec{k}}_{\regK}\ket{\vec{0}}_{\regIn}\ket{k_{i_1},\dots,k_{i_t}}_{\regSt}\ket{x_{i_1},\dots,x_{i_t}}_{R'}\\
            \ket{\phi_5} = \ket{\vec{k}}_{\regK}\ket{\vec{0}}_{\regIn}\ket{k_{i_1},\dots,k_{i_t}}_{\regSt}\left(\bigotimes_{j=1}^t V_{k_{i_j}}\ket{x_{i_j}}_{R'_j}\right)\\
            \ket{\phi_6} = \ket{\vec{k}}_{\regK}\left(\bigotimes_{j=1}^t V_{k_j}\ket{x_{j}}_{In_j}\right)\ket{\vec{0}}_{\regIn}\ket{\regK_{i_1},\dots,\regK_{i_t}}_{\regSt}\ket{\vec{0}}_{R'}
        \end{split}
    \end{equation}
    It is clear that $\ket{\phi_6} = \left(\ESim^{s,t} \ket{\phi_1}\otimes \ket{\vec{0}}_{R'}\right)$. 
\end{proof}

\begin{lemma}\label{lem:puretounpure}
    For all states $\rho_{\regK,\regIn}$ such that $Tr((\Pi_{dist,\regK}^{\lambda,t}\otimes I_{\regIn})\rho) = 1$,
    $$\Sim^t(\rho)  = \Tr_{\regSt}(\Sim^{t}_{{\sf pure}}(\rho \otimes \ketbra{\emptyset}_{\regSt}))$$
\end{lemma}

\begin{proof}
    This follows immediately from the principle of deferred measurement and the fact that measuring in the standard basis on register $\regSt$ commutes with $\Select_{\regK_i,\regIn_i,\regSt,\regR}$ and $\Apply^\U_{\regSt_i,\regR_i}$
\end{proof}

~\Cref{thm:unisimmain} then follows directly from~\Cref{lem:pureuni,lem:puretounpure,thm:isosim,thm:isosimeff}.
\else
A full proof of this theorem is deferred to~\Cref{sec:unitaryapp}.
\fi
\section{$t$-copy Pseudorandom Unitaries}

\begin{definition}
    We say that a pseudorandom unitary is pure if for all keys $k\in \{0,1\}^{\ell_k(\secparam)}$, for all pure states $\ket{\phi}$ over $\{0,1\}^{\ell_n(\secparam)}$, there exists a pure state $\ket{\psi}$ such that 
    $$\PRU_\secparam \ket{k}_{\regK}\ket{\phi}_{\regIn}\ket{0}_{\regAnc} = \ket{k}_{\regK}\ket{\psi}_{\regIn}\ket{0}_{\regAnc}$$

    Note that when $\PRU_{\secparam}$ is pure, $\PRU_k$ is a unitary. Recall the map $\Apply^{\PRU} \ket{k}\ket{x}\mapsto \ket{k}\PRU_k \ket{x}$. When $\PRU$ is pure, this map is an efficiently implementable unitary.
\end{definition}

\begin{note}
    As far as the authors are aware, all constructions of PRUs in the literature are pure~\cite{MH24}. \pnote{why almost all?}
\end{note}

\begin{theorem}\label{thm:mainunitarythm}

    Let $\PRU$ be a pure $1$-time pseudorandom unitary with keys of length $\ell_k(\secparam)$ over states of length $\ell_n(\secparam)$.
    
    Let $t(\secparam),\ell'(\secparam)$ be any polynomials such that $\ell'=\omega(\log\secparam)$ and $\ell' \leq \frac{\ell_n}{2}$. Let $\{f_k:\{0,1\}^{\ell'} \to \{0,1\}^{\ell_k}\},\{g_k:\{0,1\}^{\ell'}\to [2t]\}$ be two $\negl(\secparam)$-approximate $2t$-wise independent hash functions with keys of length $\ell_{k_f},\ell_{k_g}$ respectively. Let $\{U_k\}$ be a $\negl(\secparam)$-approximate $t$-design on $\ell'$ qubits with keys of length $\ell_{k_U}$.

    Then there exists a non-adaptive, pure, $t$-time pseudorandom unitary with keys of length $\ell'_k = \ell_{k_f}+\ell_{k_g}+\ell_{k_U}$ over states of length $\ell'+\ell_n$.
\end{theorem}

Setting $\epsilon = \frac{1}{2^{\ell_k}}$,~\Cref{thm:efftwise,thm:effdesign} give 
\begin{corollary}
    Let $\PRU$ be a pure $1$-time pseudorandom unitary with keys of length $\ell_k(\secparam)$ over states of length $\ell_n(\secparam)$. Let $t$ be any polynomial. Let $\ell'$ be any polynomial such that $\ell'=\omega(\log \secparam)$ and $\ell' \leq \ell_n/2$.

    Then there exists a non-adaptive, pure, $t$-time pseudorandom unitary with keys of length $O(t\cdot (\ell_k+\ell'))$ over states of length $\ell_n+\ell'$.
\end{corollary}

We will first introduce some information-theoretic auxiliary lemmas, which we will prove using the path-recording method introduced by~\cite{MH24}. The proofs will be deferred to~\Cref{sec:auxlemmmas}.
\ifnum\llncs=1
\begin{theorem}\label{lem:nocoll}
    Let $\regIn=\regIn_1,\dots,\regIn_t$ be a register on $\mathcal{H}(\{0,1\}^n)^{\otimes t}$ for any $t,n\in \N$. $\rho_{\regIn_1,\dots,\regIn_t}$ be any state. Then
    $$\Tr(\Pi_{{\sf dist}}^{n,t}\E_{U\gets \Haar(\{0,1\}^n)}[U^{\otimes t}\rho U^{\dagger,\otimes t}]) \geq 1 - \frac{2t^2}{2^n}$$
\end{theorem}

\begin{lemma}\label{lem:idealunitary}
    Let $\mathcal{O}_1$ be defined by the following process:
    \begin{enumerate}
        \item On initialization, sample $U\gets \Haar(\{0,1\}^\ell)$, for each $k\in \{0,1\}^\ell$, sample $U'_k\gets \Haar(\{0,1\}^n)$. 
        \item When queried on registers $\mathbf{K},\regIn$ over $\mathcal{H}(\{0,1\}^\ell),\mathcal{H}(\{0,1\}^n)$ respectively, apply $\Apply^{\{U'_k\}}_{\mathbf{K},\regIn}\cdot U_{\mathbf{K}}$
    \end{enumerate}
    Let $\mathcal{O}_2$ be a Haar random unitary. Then for all non-adaptive $t$ query quantum algorithms $\A^{(\cdot)}$,
    $$\abs{\Pr[\A^{\mathcal{O}_1} \to 1]-\Pr[\A^{\mathcal{O}_2}\to 1]} \leq 2^\ell\frac{2t^2}{2^n} + \frac{2t^2}{2^\ell} + \frac{2t^2}{2^{\ell + n}}$$
\end{lemma}

\else

\ifnum\llncs=0
\subsection{Auxiliary Lemmas}\label{sec:auxlemmmas}
\else
\section{Auxiliary Lemmas}\label{sec:auxlemmmas}
\fi

\begin{definition}
    Define $R^{inj}_{n}\subseteq \mathcal{P}(\{0,1\}^{2n})$ to be the set of injective relations over $\{0,1\}^n$. Formally,
    $$R^{inj}_n = \{R\subseteq \{0,1\}^{2n}:\forall(x,y)\neq(x',y')\in R, y\neq y'\}$$
    
    We define $R^{inj}_{n,t} \subseteq R^{inj}_n$ to be
    $$R_{n,t}^{inj} = \{R\in R^{inj}_{n,t}:|R|\leq t\}$$
\end{definition}

\begin{definition}[Forward query path recording oracle]
    Let $n\in \N$ and let $t_{max} \leq 2^n$. Let $V_{n}$ be the partial isometry over $\mathcal{H}(\{0,1\}^n)\otimes \mathcal{H}(R^{inj}_n)$ defined as follows: for $x\in \{0,1\}^n,D\in R^{inj}_{n,t_{max}-1}$
    $$V\ket{x}\ket{D} \mapsto \frac{1}{\sqrt{2^n-|Im(D)|}}\sum_{y\in \{0,1\}^n\setminus Im(D)}\ket{y}\ket{D\cup \{(x,y)\}}$$
\end{definition}

\begin{theorem}[Theorem 5~\cite{MH24}]\label{thm:pro}
    Let $\mathcal{A}^{(\cdot)}$ be any $t$ query algorithm operating on registers $AB$, where register $A$ is over $\mathcal{H}(\{0,1\}^n)$. Let $V_{n}$ operate on registers $AR$. Then,
    $$TD\left(\E_{U\gets \Haar(S)}\left[\ketbra{\A^{U}}\right] , \Tr_{R}\left(\ketbra{\A^{V_n}_{ABR}}\right)\right)\leq \frac{2t(t-1)}{2^n+1}$$
\end{theorem}

We will also use a modified version of the path recording oracle which always outputs distinct prefixes.

\begin{definition}[Modified path recording oracle]
    Let $\ell,n\in \N$ and let $t_{max} \leq 2^n$. Let $V_{\ell,n}$ be the partial isometry over $\mathcal{H}(\{0,1\}^\ell)\otimes \mathcal{H}(\{0,1\}^n)\otimes \mathcal{H}(R^{inj}_n)$ defined as follows: \pnote{is it supposed to be $D$ instead of $D_1$ below?} for $x\in \{0,1\}^n,D\in R^{inj}_{n,t_{max}-1}$
    $$V\ket{a,x}\ket{D} \mapsto \propto\sum_{\substack{b\in \{0,1\}^\ell,(b,\cdot)\notin {\sf Im}(D)\\y\in \{0,1\}^n}}\ket{b,y}\ket{D\cup \{((a,x),(b,y))\}}$$
\end{definition}

\begin{theorem}[Follows from Theorem 9~\cite{MH24}]\label{thm:promod}
    Let $\mathcal{A}^{(\cdot)}$ be any $t$ query algorithm operating on registers $AB$, where register $A$ is over $\mathcal{H}(\{0,1\}^n)$. Let $V_{n}$ operate on registers $AR$. Then,
    $$TD\left(\E_{U\gets \Haar(S)}\left[\ketbra{\A^{U}}\right] , \Tr_{R}\left(\ketbra{\A^{V_n}_{ABR}}\right)\right)\leq \frac{2t(t-1)}{2^n+1}$$
\end{theorem}

\ifnum\llncs=0
\begin{theorem}\label{lem:nocoll}
    Let $\regIn=\regIn_1,\dots,\regIn_t$ be a register on $\mathcal{H}(\{0,1\}^n)^{\otimes t}$ for any $t,n\in \N$. $\rho_{\regIn_1,\dots,\regIn_t}$ be any state. Then
    $$\Tr(\Pi_{{\sf dist}}^{n,t}\E_{U\gets \Haar(\{0,1\}^n)}[U^{\otimes t}\rho U^{\dagger,\otimes t}]) \geq 1 - \frac{2t^2}{2^n}$$
\end{theorem}
\else
\begin{theorem}[\Cref{lem:nocoll} restated]
    Let $\regIn=\regIn_1,\dots,\regIn_t$ be a register on $\mathcal{H}(\{0,1\}^n)^{\otimes t}$ for any $t,n\in \N$. $\rho_{\regIn_1,\dots,\regIn_t}$ be any state. Then
    $$\Tr(\Pi_{{\sf dist}}^{n,t}\E_{U\gets \Haar(\{0,1\}^n)}[U^{\otimes t}\rho U^{\dagger,\otimes t}]) \geq 1 - \frac{2t^2}{2^n}$$
\end{theorem}
\fi

\begin{proof}
    \pnote{not sue I follow this proof:} Let $V_n$ be the path recording oracle operating on an internal register $\regR$. Let $\ket{x_1,\dots,x_t}_{\regIn}$ be some standard basis state. Then
    $$V_{n,\regIn_1,\regR} \otimes \dots V_{n,\regIn_t,\regR} \ket{x_1,\dots,x_t}_{\regIn}\ket{\emptyset}_{\regR} = \sum_{y_1\neq \dots\neq y_t}\ket{y_1,\dots,y_t}\ket{\{(x_1,y_1),\dots,(x_t,y_t)\}}$$
    It is clear that this state is contained in ${\sf Im}(\Pi_{dist}^{n,t})$. And so in particular,
    $$\Tr((\Pi^{n,t}_{\sf dist,\regIn} \otimes I_{\regR}) (V_{n,\regIn_1,\regR} \otimes \dots V_{n,\regIn_t,\regR}) (\rho_{\regIn} \otimes \ketbra{\emptyset}_{\regR}) (V_{n,\regIn_1,\regR}^\dagger \otimes \dots V_{n,\regIn_t,\regR}^\dagger)) = 1$$
    And so the result follows from~\Cref{thm:pro}.
\end{proof}

\ifnum\llncs=0
\begin{lemma}\label{lem:idealunitary}
    Let $\mathcal{O}_1$ be defined by the following process:
    \begin{enumerate}
        \item On initialization, sample $U\gets \Haar(\{0,1\}^\ell)$, for each $k\in \{0,1\}^\ell$, sample $U'_k\gets \Haar(\{0,1\}^n)$. 
        \item When queried on registers $\mathbf{K},\regIn$ over $\mathcal{H}(\{0,1\}^\ell),\mathcal{H}(\{0,1\}^n)$ respectively, apply $\Apply^{\{U'_k\}}_{\mathbf{K},\regIn}\cdot U_{\mathbf{K}}$
    \end{enumerate}
    Let $\mathcal{O}_2$ be a Haar random unitary. Then for all non-adaptive $t$ query quantum algorithms $\A^{(\cdot)}$,
    $$\abs{\Pr[\A^{\mathcal{O}_1} \to 1]-\Pr[\A^{\mathcal{O}_2}\to 1]} \leq 2^\ell\frac{2t^2}{2^n} + \frac{2t^2}{2^\ell} + \frac{2t^2}{2^{\ell + n}}$$
\end{lemma}
\else
\begin{lemma}[\Cref{lem:idealunitary} restated]
    Let $\mathcal{O}_1$ be defined by the following process:
    \begin{enumerate}
        \item On initialization, sample $U\gets \Haar(\{0,1\}^\ell)$, for each $k\in \{0,1\}^\ell$, sample $U'_k\gets \Haar(\{0,1\}^n)$. 
        \item When queried on registers $\mathbf{K},\regIn$ over $\mathcal{H}(\{0,1\}^\ell),\mathcal{H}(\{0,1\}^n)$ respectively, apply $\Apply^{\{U'_k\}}_{\mathbf{K},\regIn}\cdot U_{\mathbf{K}}$
    \end{enumerate}
    Let $\mathcal{O}_2$ be a Haar random unitary. Then for all non-adaptive $t$ query quantum algorithms $\A^{(\cdot)}$,
    $$\abs{\Pr[\A^{\mathcal{O}_1} \to 1]-\Pr[\A^{\mathcal{O}_2}\to 1]} \leq 2^\ell\frac{2t^2}{2^n} + \frac{2t^2}{2^\ell} + \frac{2t^2}{2^{\ell + n}}$$
\end{lemma}
\fi

\begin{proof}
    We replace $U$ with a path-recording oracle $Pr=V_\ell$.
    
    We will further replace unitary $U'_k$ with a path-recording oracle $Pr_k=V_{n}$, producing a new oracle $\mathcal{O}_1'$. In particular, $\mathcal{O}'_1$ will act as follows
    \begin{equation}
    \begin{split}
        \ket{x_1,x_2}\ket{D,D_1,\dots,D_{2^\ell}}\\
        \mapsto \propto\sum_{y_1\notin D}\sum_{y_2\notin D_{y_1}}\ket{y_1,y_2}\ket{D\cup \{(x_1,y_1)\},D_1,\dots,D_i\cup \{(x_2,y_2)\},\dots,D_{2^\ell}}
    \end{split}
    \end{equation}

    Given a database $\wt{D}=\{((x_1^1,x_2^1),(y_1^1,y_2^1)),\dots,((x_1^2,x_2^2),(y_1^2,y_2^2))\}$ over $\{0,1\}^{\ell}\times \{0,1\}^n$, define $\expand(D) = (D,D'_1,\dots,D'_{2^\ell})$ to be the following
    \begin{enumerate}
        \item $D = \{(x_1^1,y_1^1),\dots,(x_1^t,y_1^t)\}$
        \item For $j$ such that $j=y_1^i$, define $D_j' = \{(x_2^i,y_2^i)\}$
        \item For all other $j$, define $D_j'=\emptyset$.
    \end{enumerate}

    Let $\mathcal{O}_2'=V_{\ell,n}$.
    
    Define the isometry $Uncompress\ket{D} \mapsto \ket{Expand(D)}$.

    By construction, we have that for all $\ket{\phi}_{A\regIn}$, $Uncompress_{\mathbf{D}}\cdot (\mathcal{O}_2')_{\regIn_1,\dots,\regIn_t,\mathbf{D}}^{\otimes t}\ket{\phi}_{A,\regIn,\mathbf{D}} = (\mathcal{O}_1')_{\regIn_1,\dots,\regIn_t,\mathbf{D}}^{\otimes t}\ket{\phi}_{A,\regIn,\mathbf{D}}$.

    Thus, since $Uncompress$ only acts on the database register, for all $t$ parallel query quantum algorithms $\A$,
    $$\Pr[\A^{\mathcal{O}_1'} \to 1]=\Pr[\A^{\mathcal{O}_2'}\to 1]$$

    The theorem then follows by~\Cref{thm:pro,thm:promod}.
\end{proof}
\fi

\subsection{Proof of~\Cref{thm:mainunitarythm}}

\begin{proof}

We will define the non-adaptive, pure, $t$-time pseudorandom unitary $\wt{PRU}$ by defining the unitaries $\wt{PRU}_{\wt{k}}$ for each key $\wt{k}$. $\wt{PRU}_{\wt{k}}$ will act on input register $\wt{\regIn} = (\regK,\regIn)$. This construction is visualized in~\Cref{fig:prucons}.

\begin{enumerate}
    \item Parse $\wt{k}$ as $(k_f,k_g,k_U)$ where $k_f\in\{0,1\}^{\ell_{k_f}},k_g\in \{0,1\}^{\ell_{k_g}},k_U\in \{0,1\}^{\ell_{k_U}}$.\pnote{I found the phrasing to be a bit confusing? $k_f$ being key for $f_k$ seems like a recursive definition? we should perhaps use a different notation.}
    \item Let $\Apply^{f_{k_f},\PRU}_{\regK,\regIn}$ be the map which sends $\Apply^{f_{k_f},\PRU}_{\regK,\regIn}\ket{r}_{\regK}\otimes \ket{\phi}_\regIn\mapsto \ket{k}_{\regK}\otimes (\PRU_{f_{k_f}(r)}\ket{\phi}_\regIn)$.
    \item Let $S^{g_{k_g}}$ be the map which sends $\ket{x}\mapsto \omega_{2t}^{g_{k_g}(x)}\ket{x}$.
    \item Define $\wt{\PRU}_{k_f,k_g,k_U}$ to act on registers $\regK',\regIn$ where $\regK'$ is over $\{0,1\}^{\ell'(\secparam)}$ and $\regIn$ is over $\{0,1\}^{\ell_n(\secparam)}$.
    \item We then define 
    $$\wt{PRU}_{k_f,k_g,k_U}\coloneqq \Apply^{f_{k_f},\PRU}_{\regK,\regIn}\cdot S^{g_{k_g}}_{\regK'} \cdot (U_{k_U})_{\regK'}$$
\end{enumerate}

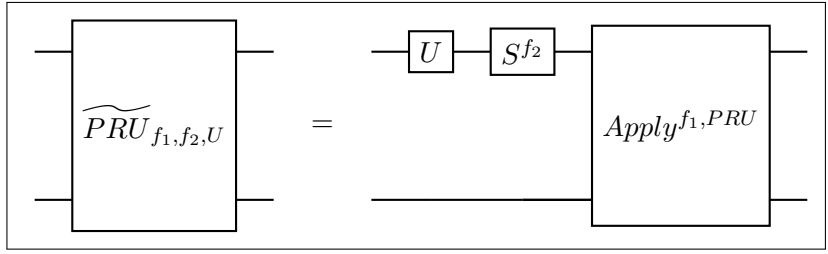
\begin{figure}
    \centering
    \fbox{
        \begin{quantikz}
        \qw&\gate[wires=3,nwires=2]{\wt{\PRU}_{f_1,f_2,U}}&\qw&  &&\gate{U}&\gate{S^{f_2}}&\gate[wires=3,nwires=2]{Apply^{f_1,\PRU}}&\qw \\
        &&&=&&&&&\\
        \qw&&\qw&  &  & \qw & \qw & \qw  & \qw
        \end{quantikz}
    }
\hfill
\caption{Construction of a $t$-copy non-adaptive PRU from a $1$-copy pure PRU. The key consists of $f_1,f_2$ two $2t$-wise independent hashes, and $U,U'$ two $t$-designs.}\label{fig:prucons}
\end{figure}

Note that since $\PRU$ is pure, $\Apply^{f_{k_f},\PRU}$ can be efficiently implemented by writing $f_{k_f}(k)$ in an ancilla register $C$, running $\PRU$ using register $C$ as the key register, and then clearing register $C$ by recomputing $f_{k_f}(k)$.

We will then show that this is a $t$-copy non-adaptive pseudorandom unitary. In particular, let $\A^{(\cdot)}$ be any $t$-query non-adaptive QPT adversary. We will model $\A$ as two efficient CPTP maps $\A_1,\A_2$ such that
$\A^{\Phi}=(\A_2 \circ \Phi \circ \A_1)(\ketbra{0})$. The role of $\Phi$ will be clear later. 

We then proceed to show that this construction is secure. We will do this via a sequence of hybrids. In particular, we will define a sequence of oracles defined by CPTP $\Phi_1,\dots,\Phi_4$. We will then show that for all non-adaptive $t$-query QPT oracle algorithms $\A=\A_1,\A_2$ and for all $i$,
$$\abs{\Pr[\A^{\Phi_{i}}\to 1] - \Pr[\A^{\Phi_{i+1}} \to 1]}\leq \negl(\secparam)$$
Here $\Phi_1$ will represent non-adaptive queries to $\wt{\PRU}$, while $\Phi_4$ will represent non-adaptive queries to a Haar random unitary. 
\par We present the following hybrids.

\noindent $\hybrid_1$: $\Phi_1$ will be $\wt{PRU}^{\otimes t}_{\wt{k}}$ for a random key $\wt{k}=(k_f,k_g,k_U)$.\\

\noindent $\hybrid_2$: $\Phi_2$ will be the same as $\Phi_1$, but with $U_{k_U}$ replaced by a Haar random unitary $U$ and with $f_{k_f},g_{k_g}$ replaced by random functions $f,g$.\\
    
\noindent $\hybrid_3$:  $\Phi_3$ will be the same as $\Phi_2$, but will project onto $\Pi_{dist}$ after applying the first round of $U$'s. Formally, on input $\rho_{\regK,\regIn}$, it will do the following
    \begin{enumerate}
        \item Sample $U\gets \Haar(\{0,1\}^\ell)$, $f:\{0,1\}^{\ell'}\to \{0,1\}^{\ell_n},g:\{0,1\}^{\ell'}\to [2t]$ random functions.
        \item Apply $U_{\regK_1}\otimes \dots \otimes U_{\regK_t}$.
        \item Apply the measurement $\{\Pi_{dist},I-\Pi_{dist}\}$ on $\regK_1,\dots,\regK_t$. If the result is the second option, output $\bot$.
        \item Otherwise, apply $(\Apply^{f,\PRU}_{\regK_1,\regIn_1}\cdot S^{g}_{\regK_1})\otimes \dots \otimes (\Apply^{f,\PRU}_{\regK_t,\regIn_t}\cdot S^{g}_{\regK_t})$.
        \item Finally, output registers $\regK,\regIn$.
    \end{enumerate}
\ \\
\noindent $\hybrid_4$: Define $\PRU^{f} = \{\PRU^{f}_k\}_{k\in \{0,1\}^{\secparam}}$ to be the family of unitaries defined by $\PRU^{f}_k = \PRU_{f(k)}$. Let $\Sim^{t,\PRU^{f}}$ be the simulator from~\Cref{thm:unisimmain} instantiated with the family $\PRU^{f}$. $\Phi_4$ will act as follows
    \begin{enumerate}
        \item Sample $U\gets \Haar(\{0,1\}^\ell)$, $f:\{0,1\}^{\ell'}\to \{0,1\}^{\ell_n}$ a random function.
        \item Apply $U_{\regK_1}\otimes \dots \otimes U_{\regK_t}$.
        \item Apply the measurement $\{\Pi_{dist},I-\Pi_{dist}\}$ on $\regK_1,\dots,\regK_t$. If the result is the second option, output $\bot$.
        \item Otherwise, apply $\Sim^{t,\PRU^{f}}_{\regK,\regIn}$.
        \item Finally, output registers $\regK,\regIn$.
    \end{enumerate}
\ \\
\noindent $\hybrid_5$: $\Phi_5$ will be defined as $\Phi_4$ with the following modification. Whenever the simulator queries $\PRU_{f(r)}$ on any (classical) input $r$, $\Phi_5$ will instead pick a fresh $r'$ uniformly at random and run $\PRU_{r'}$.\\
\ \\
\noindent $\hybrid_6$:  $\Phi_6$ will be the same as $\Phi_5$ with the following modification. Whenever the simulator queries $\PRU_{f(r)}$, it will instead sample a fresh Haar random unitary and apply that.\\
    
\noindent $\hybrid_7$: $\Phi_7$ will be the same as $\Phi_6$, but with the simulator replaced by a new simulator $\Sim^{t,\{U'_k\}}$ for a freshly sampled family of Haar random unitaries $\{U'_k\}$. Formally,    
    \begin{enumerate}
        \item Sample $U\gets \Haar(\{0,1\}^\ell)$.
        \item For each $k\in \{0,1\}^\ell$, sample $U'_k\gets \Haar(\{0,1\}^{n})$. 
        \item Apply $U_{\regK_1}\otimes \dots \otimes U_{\regK_t}$.
        \item Apply the measurement $\{\Pi_{dist},I-\Pi_{dist}\}$ on $\regK_1,\dots,\regK_t$. If the result is the second option, output $\bot$.
        \item Otherwise, apply $\Sim^{t,\{U'_k\}}_{\regK,\regIn}$.
        \item Finally, output registers $\regK,\regIn$.
    \end{enumerate}
\ \\ 
\noindent $\hybrid_8$:  $\Phi_8$ will be the same construction as $\Phi_2$, but with $\PRU$ replaced by a family of Haar random unitaries. Formally,
    \begin{enumerate}
        \item Sample $U\gets \Haar(\{0,1\}^\ell)$, $g:\{0,1\}^{\ell'}\to [2t]$ a random function.
        \item For each $k\in \{0,1\}^\ell$, sample $U'_k\gets \Haar(\{0,1\}^{n})$. 
        \item Apply $U_{\regK_1}\otimes \dots \otimes U_{\regK_t}$.
        \item Apply the measurement $\{\Pi_{dist},I-\Pi_{dist}\}$ on $\regK_1,\dots,\regK_t$. If the result is the second option, output $\bot$.
        \item Otherwise, apply $(\Apply^{\{U'_k\}}_{\regK_1,\regIn_1}\cdot S^{g}_{\regK_1})\otimes \dots \otimes (\Apply^{\{U'_k\}}_{\regK_t,\regIn_t}\cdot S^{g}_{\regK_t})$.
        \item Finally, output registers $\regK,\regIn$.
    \end{enumerate}
\ \\    
\noindent $\hybrid_9$: $\Phi_9$ will be the same as $\Phi_8$ but with the application of $\Pi_{dist}$ removed.\\
    
\noindent $\hybrid_{10}$: $\Phi_{10}$ will be the same as $\Phi_9$, but with the application of $S^{g}$ removed. Formally,
    \begin{enumerate}
        \item Sample $U\gets \Haar(\{0,1\}^\ell)$.
        \item For each $k\in \{0,1\}^\ell$, sample $U'_k\gets \Haar(\{0,1\}^{n})$. 
        \item Apply $U_{\regK_1}\otimes \dots \otimes U_{\regK_t}$.
        \item Otherwise, apply $\Apply^{\{U'_k\}}_{\regK_1,\regIn_1}\otimes \dots \otimes \Apply^{\{U'_k\}}_{\regK_t,\regIn_t}$.
        \item Finally, output registers $\regK,\regIn$.
    \end{enumerate}
\ \\   
\noindent $\hybrid_{11}$: Finally, $\Phi_{11}$ will be a $t$-fold Haar random unitary. \\

\noindent We show the indistinguishability of every pair of consecutive hybrids below.

\begin{claim}$\abs{\Pr[\A^{\Phi_{1}}\to 1] - \Pr[\A^{\Phi_{2}} \to 1]}\leq \negl(\secparam)$\end{claim}
    \begin{proof}This follows directly from the fact that $\{U_k\}$ is a negligibly approximate $t$-design and that $\{f_k\},\{g_k\}$ are $2t$-wise independent hash functions (applying~\Cref{lem:twise}).
    \end{proof}
    
 \begin{claim}$\abs{\Pr[\A^{\Phi_{2}}\to 1] - \Pr[\A^{\Phi_{3}} \to 1]} \leq \negl(\secparam)$\end{claim}
    \begin{proof}
    This follows from~\Cref{lem:nocoll} and gentle measurement. In particular, we know that the measurement $\{\Pi_{dist},I-\Pi_{dist}\}$ will output the first result with all but negligible probability, and so by gentle measurement, performing this measurement can have at most a negligible impact on the resulting output probability.
    \end{proof}
    
 \begin{claim}
    $\abs{\Pr[\A^{\Phi_{3}}\to 1] - \Pr[\A^{\Phi_{4}} \to 1]} = 0$
    \end{claim} 
    \begin{proof}
    This follows directly from~\Cref{thm:unisimmain}.
    \end{proof}
    
 \begin{claim}$\abs{\Pr[\A^{\Phi_{4}}\to 1] - \Pr[\A^{\Phi_{5}} \to 1]} = 0$\end{claim} \begin{proof} This follows from the fact that the simulator defined in~\Cref{thm:unisimmain} explicitly queries $\PRU^{f}$ on $t$ distinct classical inputs only once. And so by lazy sampling, it is equivalent to sample the values of $f(\cdot )$ when they are first queried.\end{proof}
    
 \begin{claim}$\abs{\Pr[\A^{\Phi_{5}}\to 1] - \Pr[\A^{\Phi_{6}} \to 1]} \leq \negl(\secparam)$\end{claim} \begin{proof}This follows directly from the fact that $\PRU$ is a $1$-copy pseudorandom unitary, since it is only queried directly on random keys and once for each key.\end{proof}
    
 \begin{claim} $\abs{\Pr[\A^{\Phi_{6}}\to 1] - \Pr[\A^{\Phi_{7}} \to 1]} = 0$\end{claim} \begin{proof} Note that in $\Phi_7$, each $U'_k$ is Haar random and queried at most once. Thus, it is equivalent to sample $U'_k$ only at the point when it is queried.\end{proof}
    
\begin{claim} $\abs{\Pr[\A^{\Phi_{7}}\to 1] - \Pr[\A^{\Phi_{8}} \to 1]} \leq \negl(\secparam)$ \end{claim} \begin{proof} This follows directly from~\Cref{thm:unisimmain}.\end{proof}
    
\begin{claim} $\abs{\Pr[\A^{\Phi_{8}}\to 1] - \Pr[\A^{\Phi_{9}} \to 1]} \leq \negl(\secparam)$\end{claim} 
    \begin{proof} This follows from~\Cref{lem:nocoll}.\end{proof}
    
 \begin{claim} $\abs{\Pr[\A^{\Phi_{9}}\to 1] - \Pr[\A^{\Phi_{10}} \to 1]} = 0$\end{claim} \begin{proof} This follows from unitary invariance. In particular, $\Apply^{\{U'_k\}}\cdot S^{f_2} = \Apply^{\{\omega_q^{f_2(k)}U'_k\}}$, and by unitary invariance the distribution $\{\omega_q^{f_2(k)}U'_k\}$ is identically distributed to $\{U'_k\}$.\end{proof}
    
 \begin{claim} $\abs{\Pr[\A^{\Phi_{10}}\to 1] - \Pr[\A^{\Phi_{11}} \to 1]} \leq \negl(\secparam)$\end{claim}
    \begin{proof} This follows from~\Cref{lem:idealunitary}.\end{proof}

\end{proof}

\ifnum\llncs=0
\section*{Acknowledgements}
PA would like to thank Zihan Hu for enlightening preliminary discussions on designing multi-copy secure public-key quantum money schemes. PA is supported by the National Science Foundation under the grants FET-2329938, CAREER-2341004 and, FET-2530160. EG is supported by the National Science Foundation Graduate Research Fellowship Program.
\par We acknowledge the use of generative AI tools for improving the presentation quality of the paper. However, they were not used for deriving any results in this paper.  
\fi
\ifnum\llncs=0
\bibliographystyle{alpha}
\bibliography{references}
\fi 
\ifnum\llncs=1
\section*{Acknowledgements}
\noindent We acknowledge the use of generative AI tools for improving the presentation quality of the paper. However, they were not used for deriving any results in this paper.  
\bibliographystyle{plain}
\bibliography{references}
\fi

\ifnum\llncs=1
\newpage 
\begin{center}
	{\bf{\Large Supplementary Material}}
	\end{center}
\appendix

\section{Details on the Unitary Setting}\label{sec:unitaryapp}

\newpage 
\section{Sticky Reviews}
\newcommand{\response}[1]{{\color{blue} {\bf Response}: #1}}
\begin{verbatim}
This paper was previously submitted to EUROCRYPT 2026.

The paper has been revised according to the feedback from the referees.

We state and address the relevant referee’s comments below.

\end{verbatim}

\noindent {\bf Referee said}: In the proof of the main theorem (Theorem 11), I had difficulty understanding the description of the simulator. The proof states that the simulator takes as input eigenbases sampled uniformly at random. However, the spectral decompositions of given mixed states
 do not necessarily have uniform eigenvalues. Does this part mean that the eigenvectors are sampled according to their corresponding eigenvalues?\\

 \noindent \response{We agree, the eigenvectors are sampled according to their responses. On page 21, in the earlier version it said that $\{\ell_i\}$ are picked uniformly at random. We have now corrected the typo to say that $\{\ell_i\}$ are sampled from $\distr_{h_i}$. }\\

 \noindent {\bf Referee said}: There are several places where explanations are insufficient or terminology is used inconsistently. 
 \begin{enumerate}
     \item For example, in Theorem 11, the sampling procedure for $(\chi_1,...\chi_t)$ is not defined.
     \item In page 16, wt is not defined.
     \item In the proof of Theorem 11, $n=n_A$?
     \item In the results on PRS and PRU, additional assumptions are required, and it would be helpful if the paper could discuss in more detail how mild these assumptions are. In the case of PRU, it seems that existing constructions satisfy the required assumption, however,
 since the existing constructions achieve adaptive multi-query secure PRU, while this work focuses on 1-copy (or non-adaptive bounded-copy) secure PRU, I feel there is a gap between the two settings.
 \item  In the introduction, the notions of 1-copy secure PRS and
 PRU are introduced, but in the main body the authors seem to use the term 1-time PRS (and PRU) to refer to the same primitives.
 \item In addition, in the definitions of quantum money (Definitions 17–19), several algorithms should be specified as quantum polynomial-time,
 but they are currently described as probabilistic polynomial-time. These issues somewhat detract from the overall clarity of the paper.\\
 
 \end{enumerate}  

 \noindent \response{We have made the corrections. More details below:  
 \begin{enumerate}
 \item $(\chi_1,...\chi_t)$ are eigenvectors in the spectral decomposition of the density matrix $\Tr_{\regB}(\ketbra{\phi_h}_{\regA \regB})$ and are sampled according to $\distr_h$ (see page 19).
\item  ${\sf wt}$ is defined in~\Cref{sec:prelims}. 
\item For the results of PRS and PRUs, no computational assumptions are required. However, we do assume that the single-copy PRS (with short keys) has bounded amount of ancillae and the single-query PRU (with short keys) construction is pure: both these caveats are stated in the introduction as well as the technical sections. 
\item We believe there is a misunderstanding here: there is a simple construction of 1-query secure PRU that {\bf doesn't} achieve multi-query security. The construction of an $n$-qubit 1-query secure PRU is as follows: on input $k \in \{0,1\}^{\secparam}$, apply a pseudorandom generator on $k$ to stretch to $2n$ bits (where $n > \secparam$) and then apply a 1-design (random Pauli) using the $2n$ bits of pseudorandomness. This simple construction satisfies our constraints that the construction is pure. 
\item We have addressed the inconsistency between 1-time and single-copy. 
\item We have corrected the typos in the quantum money definitions. 
 \end{enumerate}
 }

 \noindent {\bf Referee said}: A strength of this paper is that it studies a very fundamental and interesting question: How the number of copies of quantum states affects the security of cryptographic primitives.
That said, due to certain technical reasons, the applicability of the main theorem appears somewhat restricted and requires additional assumptions.
I also felt that there are several points in the manuscript that require clearer or more thorough explanations, which makes the paper somewhat harder to follow.\\

 \noindent \response{We disagree with the assertion that the applicability of the main theorem is limited. Our work already shows many applications:
 \begin{itemize}
     \item Single-copy PRS (with bounded ancillae) to multi-copy PRS transformation. Single-query PRU (with a unitary construction) to multi-query secure PRU. These transformations are {\bf new}. Prior to our work, the notions of single-copy PRS and multi-copy PRS had been extensively studied. However, despite many years, their relationship was not well understood and our work sheds light on their relationship for the first time. 
     \item The constructions of identical-copy secure quantum money and copy-protection (two of the foundational unclonable cryptographic notions) are also {\bf new}. 

 \end{itemize}
 In fact, after posting our paper, a few concurrent and independent papers were also posted. In particular, the purification framework (see~\Cref{sec:relatedwork}), which is very related to our main theorem, has led to several new applications in quntum information theory over the last few months.\\
 
 \noindent We have also made efforts to add more explanation at places to make it more accessible.}

\fi

\end{document}